%% file: main.tex
\PassOptionsToPackage{nameinlink}{cleveref}
\documentclass[a4paper,UKenglish,cleveref, autoref, thm-restate,numberwithinsect]{lipics-v2021}
\hideLIPIcs  %uncomment to remove references to LIPIcs series (logo, DOI, ...), e.g. when preparing a pre-final version to be uploaded to arXiv or another public repository

\newtheorem{question}[theorem]{Question}
\newtheorem{fact}[theorem]{Fact}

\numberwithin{equation}{section}
\numberwithin{figure}{section}

\input{macros}

\title{Generalised Quantifiers Based on Rabin-Mostowski Index} %TODO Please add

\author{Denis Kuperberg}{CNRS, LIP, Plume, ENS Lyon, France}{denis.kuperberg@ens-lyon.fr}{https://orcid.org/0000-0001-5406-717X}{ANR ReCiProg}
\author{Damian Niwiński}{Institute of Informatics, University of Warsaw, Poland}{niwinski@mimuw.edu.pl}{https://orcid.org/0000-0002-1342-9805}{National Science Centre, Poland (grant no.\@ 2024/55/B/ST6/00318)}
\author{Paweł Parys}{Institute of Informatics, University of Warsaw, Poland}{parys@mimuw.edu.pl}{https://orcid.org/0000-0001-7247-1408}{National Science Centre, Poland (grant no.\@ 2024/55/B/ST6/00318)}
\author{Michał Skrzypczak}{Institute of Informatics, University of Warsaw, Poland}{mskrzypczak@mimuw.edu.pl}{https://orcid.org/0000-0002-9647-4993}{National Science Centre, Poland (grant no.\@ 2024/55/B/ST6/00318)}

\authorrunning{D. Kuperberg, D. Niwiński, P. Parys, and M. Skrzypczak} %TODO mandatory. First: Use abbreviated first/middle names. Second (only in severe cases): Use first author plus 'et al.'

\Copyright{Denis Kuperberg, Damian Niwiński, Paweł Parys, and Michał Skrzypczak} %TODO mandatory, please use full first names. LIPIcs license is "CC-BY";  http://creativecommons.org/licenses/by/3.0/

\ccsdesc[100]{Theory of computation~Automata over infinite objects}
\ccsdesc[100]{Theory of computation~Logic and verification}
\ccsdesc[100]{Theory of computation~Tree languages}

\keywords{monadic quantifiers, decidability, quantifier elimination, parity automata, game quantifier, Rabin-Mostowski index} %TODO mandatory; please add comma-separated list of keywords

\category{} %optional, e.g. invited paper

\relatedversion{} %optional, e.g. full version hosted on arXiv, HAL, or other respository/website
\nolinenumbers %uncomment to disable line numbering

\newcommand{\topn}{\mathit{top}}
\newcommand{\len}{\mathit{len}}
\newcommand{\lk}{\mathsf{lk}}
\newcommand{\add}{\mathsf{add}}

\begin{document}

\maketitle

\begin{abstract}
	In this work we introduce new generalised quantifiers which allow us to express the Rabin\=/Mostowski index of automata.
	Our main results study expressive power and decidability of the monadic second\=/order (MSO) logic extended with these quantifiers.
	We study these problems in the realm of both $\omega$\=/words and infinite trees.
	As it turns out, the pictures in these two cases are very different.
	In the case of $\omega$\=/words the new quantifiers can be effectively expressed in pure MSO logic.
	In contrast, in the case of infinite trees, addition of these quantifiers leads to an~undecidable formalism.

	To realise index\=/quantifier elimination, we consider the extension of MSO by game quantifiers. As a~tool, we provide a~specific quantifier\=/elimination procedure for them. Moreover, we introduce a~novel construction of~transducers realising strategies in $\omega$\=/regular games with monadic parameters.
\end{abstract}
\newpage

\section{Introduction}

Monadic second\=/order logic (MSO) considered over $\omega$\=/words or infinite trees sets a~golden standard in the theory of verification as a~robust, expressive, yet still decidable formalism. The research surrounding this logic
%has focused mostly on two aspects.
often takes two paths.

% The first
One
%is to study
focuses on
properties of the MSO\=/definable languages of $\omega$\=/words or trees, with an emphasis on decidability issues,
aiming in effective characterisations.
%leading to the term ``effective characterisation''.
% The second
Another path, maybe
%aspect, generally
more challenging, attempts to extend the expressive power of MSO while still maintaining
% its
decidability.
These two paths often interplay,
an~archetypal example being
% of the interplay between these aspects was motivated by the question of
the study of
cardinality.
% of languages.
First, Niwiński~\cite{niwinski_cardinality} showed that the cardinality of a regular language of infinite trees can be effectively computed. Then, B\'ar\'any, Kaiser, and Rabinovich~\cite{barany_expressing_trees} (see also~\cite{kaiser_automatic})
% gave attention to the \emph{cardinality quantifiers} in
studied an extension of the
MSO logic (over the binary tree) by \emph{cardinality quantifiers}, like $\exists^{\geq \kappa} X.\,\varphi(\vec{W}, X)$, stating that there are at least $\kappa$ distinct sets $X$ satisfying $\varphi(\vec{W}, X)$.
The extension turned out to admit an~elimination procedure for cardinality quantifiers:
the authors effectively translated MSO with cardinality quantifiers into pure MSO,
rendering the considered formalism decidable~\cite{barany_expressing_trees}.
%Their main result states
%that MSO with cardinality quantifiers can be into pure MSO,
%In this light, the results of~\cite{niwinski_cardinality} show decidability of formulae with a~single outermost cardinality quantifier and pure MSO inside.

%Another example of such two\=/pronged approach is
%An example of a
In contrast,
the \emph{unboundedness} quantifier $\qU X.\, \varphi(\vec{W}, X)$
% has been studied
% first studied
introduced
by Bojańczyk~\cite{bojanczyk_bounding},
stating that the formula $\varphi(\vec{W}, X)$ is satisfied
by finite sets $X$ of unbounded size,
leads to a proper extension of MSO. After exhaustive investigation it was shown that MSO+$\qU$ is undecidable even over $\omega$\=/words~\cite{bojanczyk_msou_final}. However, the unboundedness property of a~given regular language is easily decidable (due to an~application of the pumping lemma);
a~related property called diagonality was shown to be decidable
even for tree languages on all levels of the Caucal hierarchy~\cite{lorenzo-diagonal}.
%in the Caucal hierarchy
%in fact even a stronger property, called diagonal \TODO was shown to be decidable over arbitrary structures from the the Caucal hierarchy~CITE.
%\dn{How about it?}

The results of~Niwiński, Parys, and Skrzypczak~\cite{dichotomy-arxive}
fall into a similar category:
the authors show that the ranks of
MSO-definable well-founded relations
%a language or a relation
satisfy a certain dichotomy and can be effectively bounded,
although the rank itself is not
%while not being
directly expressible in MSO.
%(btw: is MSO+rank decidable?).

A general pattern behind these situations consists of
% decidability questions on
several levels.
On the basic level, we wish
to decide if a language of $\omega$\=/words or trees satisfies a specific property, usually related to some {\em difficulty\/}: uncountability, unboundedness, ordinal rank $\omega_1$, etc. Then we ask if the property can be generalised to a type of quantifier, and whether the extension of MSO is proper, and eventually decidable.
%Note that the frontier between decidability and undecidability can pass in various places.

The study in the present paper is motivated by
the Rabin-Mostowski index problem, which is a pertinent open problem in automata theory. In terms of parity automata (see below), the question is to find an equivalent automaton of a given type (deterministic, non\=/deterministic, or alternating) with a minimal number of priorities. For technical reasons, we also take into account the minimal priority, so that an {\em index\/} is defined as a pair $(i,j)$ (where $i$ can be assumed to be $0$ or $1$). Recall that the index hierarchy over $\omega$\=/words is strict only for deterministic automata, and collapses to the second level for non\=/deterministic and alternating ones. For infinite trees, both non\=/deterministic and alternating hierarchies are strict~\cite{bradfield_simplifying,niwinski_nondet_strict}; the deterministic hierarchy is strict as well, but less interesting because deterministic tree automata do not capture all regular tree languages. The problem of~computing the index is generally decidable for automata over $\omega$\=/words~\cite{wagner_hierarchy}, and open for automata over infinite trees. Several special cases have been shown decidable, in particular if an~input tree automaton is a~deterministic automaton~\cite{niwinski_gap,niwinski_deterministic}; a~\emph{game automaton}~\cite{murlak_game_auto}; or a~B\"uchi automaton~\cite{colcombet_weak,walukiewicz_buchi}. Colcombet and L\"oding~\cite{loding_index_to_bounds} reduced the non\=/deterministic index problem to a~question on asymptotic behaviour of counter automata; their paper brought a~bunch of interesting ideas (in particular, guidability), but the original problem has remained unsolved.

In the current paper, we approach the index problem ``from above'', that is, we introduce a~class of quantifiers corresponding to the index property. Using the correspondence between sets (or tuples thereof) and their characteristic functions (i.e.,~labelled infinite words or trees), a~general form of the new quantifier is
%\begin{eqnarray}\label{kwantyfikator}
\[ \qI^{D}_\anIndex X.\,\varphi(W_1,\ldots,W_k,X)\]
%\end{eqnarray}
where $D$ refers to the type of involved automata (deterministic or non\=/deterministic), and $\anIndex$ determines the index.
Such a~formula holds for a~valuation $\word{w}_1,\ldots,\word{w}_k$ if there exists an~automaton~$\CA$ of type $D$ and index $\anIndex$,
such that for every $\word{x}$ the formula $\varphi(\word{w}_1,\ldots,\word{w}_k,\word{x})$ holds if and only if $\CA$~accepts $\langle\word{w}_1,\ldots,\word{w}_k,\word{x}\rangle$. Note that in the above only $\word{x}$ varies while the $\word{w}_i $'s remain fixed, playing the role of parameters.

\ignore{Our main results are twofold: we show that MSO+$\qI$ effectively reduces to pure MSO over $\omega$\=/words; while MSO+$\qI$
\ignore{(even MSO+$\qI_{\mathrm{safety}}$)} is undecidable over infinite trees. This last theorem is up to our knowledge the first negative result about decidability of index\=/related problems over infinite trees.
In fact, we show undecidability already for the quantifier $\qI_{\mathrm{safety}}$, referring to automata that simply avoid some rejecting states. This is in contrast with the fact that the question whether a regular tree language can be recognised by a safety automaton can be easily seen to be decidable, as it corresponds to the closedness in the standard topology over infinite trees (see, e.g., \cite{loding_index_to_bounds}).
\ignore{Thus, the result may be interpreted as opening the possibility that the index problem itself will turn out to be undecidable.}}

Our main results are twofold. First, we show that MSO+$\qI$ effectively reduces to pure MSO over $\omega$-words. Second, we prove that MSO+$\qI$ is undecidable over infinite trees. To the best of our knowledge, this is the first negative decidability result for index\=/related problems over infinite trees.
In fact, we establish undecidability already for the quantifier $\qI_{\mathrm{safety}}$, which refers to automata that merely avoid some designated rejecting states. This stands in sharp contrast to the fact that deciding whether a regular tree language can be recognised by a safety automaton is straightforward, as it amounts to checking closedness in the standard topology on infinite trees (see, e.g.,~\cite{loding_index_to_bounds,loding_hab}).

To achieve the positive part of our results, namely index\=/quantifier elimination over $\omega$\=/words, we rely on a~variant of Wadge games for the index hierarchy~\cite{loding_wadge_dpda,wadge_phd}. These games can naturally be expressed in MSO equipped with \emph{game quantifier} $\qG$ (see, e.g.,~the monograph by Moschovakis~\cite{moschovakis_inductive}). The fact that MSO+$\qG$ reduces to pure MSO follows from Kaiser~\cite{kaiser_game_quant} (we provide a~direct proof adapted to our setup for the sake of completeness); nevertheless, we need a~stronger property, allowing us to construct finite memory strategies (relating B\"uchi\=/Landweber construction~\cite{buchi_synthesis} with uniformisations~\cite{lifsches_skolem,rabinovich_decidable}). This falls in similar lines as results by Winter and Zimmermann~\cite{winter-delay-games} and others on sequential uniformisation and functions realised by transducers. To achieve our goal, we show a~novel fact, which can be seen as a~parametrised version of B\"uchi\=/Landweber construction (for the case when the variables are in some sense separated). We believe that both game quantifiers in general, and this new fact are of independent interest and applications.

One can ask if the new quantifiers of our paper align with the concept of \emph{generalised quantifiers} introduced by Mostowski~\cite{mostowski-quantifiers} (see~\cite{sep-generalized-quantifiers} for a survey).
The idea there is that a~formula $\qQ x.\, \varphi (\vec{w}, x)$ expresses the fact that the $x$'s satisfying $\varphi (\vec{w}, x)$
(for fixed parameters~$\vec{w}$) fall into a~specified family of subsets of the universe (e.g.,~all non\=/empty sets for $\exists$,
and the singleton of the whole universe for $\forall$). More generally, a~quantifier can bind $k$ variables
($\qQ x_1 \ldots x_k. \, \varphi (\vec{w}, \vec{x})$) and relate to a family of $k$\=/ary relations. These concepts can be adapted to MSO, where in the semantics of a~quantifier $\qQ X$ (or $\qQ \vec{X} $), the universe is replaced by its powerset. The examples mentioned above, namely cardinality quantifiers and the unboundedness quantifier, can be easily presented in this way. The newly introduced index quantifiers and game quantifiers can as well be presented as generalised quantifiers%, but in a more subtle way, using relations of higher arities (over the powerset)
. For an~interested reader, we discuss this issue in more detail in \cref{sec:generalised}.

\section{Preliminaries}
\label{sec:preliminaries}

An~alphabet $\albet$ is a~finite non\=/empty set of symbols.
As~usual, by $\albet^\ast$ we denote the set of finite words over $\albet$, by $\albet^{+}$ the set of non\=/empty finite words over $\albet$,
and by $\albet^\omega$ the set of $\omega$\=/words over $\albet$, that is,
functions from $\N=\{0,1,2,\ldots\}$ to $\albet$. The empty word is denoted $\epsilon\in\albet^\ast$ and concatenation of two words $u$, $v$ is denoted by $u\cdot v$.
Given $n\in\N$ and either a~finite word with at least $n$ symbols or an~$\omega$\=/word $\word{w}=w_0w_1w_2\cdots$ by $\word{w}\restr_n$ we denote the finite word $w_0w_1\cdots w_{n-1}$, that is,
$\word{w}$ restricted to the first $n$ symbols.
An~$\omega$\=/word of the form $x\cdot y\cdot y\cdot y\cdot\ldots$ for some finite words $x,y\in\albet^{+}$ is called \emph{ultimately periodic}.
The prefix order on words is denoted by ${\preceq}$, with $\word{w}\preceq \word{w}'$ if there exists $n\in\N$ such that $\word{w}=\word{w}'\restr_n$.

A~(full, infinite, binary) tree over an~alphabet $\albet$ is any function $t\from \{\dL,\dR\}^\ast\to\albet$;
here a~word in $\set{\dL,\dR}^\ast$ describes a~path from the root $\epsilon$ to a~node $x\in\{\dL,\dR\}^\ast$, with $\dL$ being the left child and $\dR$ the right child. The \emph{label} of such a~node is $t(x)\in\albet$. The set of all such trees is denoted $\trees_\albet$.

We use the standard terms to navigate within a~tree, in particular $x$ is a~\emph{descendant} of $y$ if $x\succeq y$. In an~analogous way we use the terms \emph{ascendant}, \emph{parent}, and \emph{sibling}.

Subsets $L\subseteq\albet^\ast$, $L\subseteq\albet^\omega$, or $L\subseteq \trees_\albet$ are called \emph{languages}.

\subparagraph*{Transducers.}

In this work we use (sequential, deterministic, finite\=/memory) transducers from one alphabet to another.
Assume that $\albet_W$, $\albet_Y$ are some alphabets.
A~transducer $\tau$ from $\albet_W$ to $\albet_Y$ (denoted $\tau\from \albet_W \tto \albet_Y$) is a~tuple $\tau=\langle \albet_W, \albet_Y, Q_\tau, \iota_\tau, \delta_\tau\rangle$, where:
\begin{itemize}
\item $Q_\tau$ is a~finite set of \emph{states},
\item $\iota_\tau\in Q_\tau$ is the \emph{initial state},
\item $\delta_\tau\from Q_\tau\times\albet_W\to \albet_Y\times Q_\tau$ is the \emph{transition function}.
\end{itemize}
Given an~\emph{input $\omega$\=/word} $\word{w}=w_0w_1w_2\cdots\in(\albet_W)^\omega$ we inductively define the \emph{run} $\word{\rho}\eqdef\rho_0\rho_1\rho_2\cdots\in Q^\omega$
and the \emph{output $\omega$\=/word} $\tau(\word{w})\eqdef y_0y_1y_2\cdots\in(\albet_Y)^\omega$ taking $\rho_0\eqdef \iota_\tau$ and $(y_n,\rho_{n+1}) \eqdef \delta_\tau(\rho_n, w_n)$ for all $n\in\N$.

Given two transducers $\tau\from \albet_W\tto \albet_Y$ and $\tau'\from \albet_Y\tto \albet_Z$ it is easy to construct the \emph{composition} of the two,
namely a~transducer $\theta\from \albet_W\tto \albet_Z$ such that for every $\word{w}\in(\albet_W)^\omega$ we have $\theta(\word{w})=\tau'\big(\tau(\word{w})\big)$.

\subparagraph*{Parity indices.}

Assume that $i,j\in\N$ are natural numbers with $i\leq j$. The \emph{(strong) parity index} $P_{i,j}$ and the \emph{weak parity index} $W_{i,j}$ are defined by the languages
\begin{align*}
P_{i,j}&\eqdef \{k_0k_1k_2\ldots \in \{i,i{+}1,\ldots,j\}^\omega\mid \limsup_{n\to\infty} k_n \equiv 0\ \mathrm{mod}\ 2\},\displaybreak[0]\\
W_{i,j}&\eqdef \{k_0k_1k_2\ldots \in \{i,i{+}1,\ldots,j\}^\omega\mid \sup_{n\in\N} k_n \equiv 0\ \mathrm{mod}\ 2\}.
\end{align*}

An~\emph{index} is a~pair $\anIndex=\langle\albet_\anIndex,L_\anIndex\rangle$ that is either $\CP_{i,j}=\langle\set{i,i{+}1,\ldots,j},P_{i,j}\rangle$
or $\CW_{i,j}=\langle\set{i,i{+}1,\ldots,j},W_{i,j}\rangle$ for some $i,j\in\N$ with $i\leq j$.

The typical names for indices are: \emph{B\"uchi} for $\CP_{1,2}$ (infinitely many times priority~$2$), \emph{co\=/B\"uchi} for $\CP_{0,1}$ (finitely many times priority~$1$), \emph{safety} for $\CW_{0,1}$ (reaching priority~$1$ implies that we reject), and \emph{reachability} for $\CW_{1,2}$ (reaching priority~$2$ implies that we accept).

\subparagraph*{Automata over $\omega$\=/words.}

A \emph{non\=/deterministic parity $\omega$\=/word automaton} over an~alphabet~$\albet$ and of~index~$\anIndex=\langle\albet_\anIndex,L_\anIndex\rangle$ is a~tuple $\CD=\langle \albet, \anIndex, Q_\CD, \iota_\CD, \delta_\CD\rangle$, where:
\begin{itemize}
\item $Q_\CD$ is a~finite set of \emph{states},
\item $\iota_\CD\subseteq Q_\CD$ is the set of \emph{initial states},
\item $\delta_\CD\subseteq Q_\CD\times\albet\times \albet_\anIndex\times Q_\CD$ is the \emph{transition relation},
\end{itemize}
and moreover the automaton is \emph{complete}\footnote{This technical assumption plays a~role when considering weak indices of automata.}
in the sense that for every $q\in Q_\CD$ and $a\in\albet$ there is at least one transition of the form $(q,a,k,q')\in\delta_\CD$.

A~\emph{run} of an~automaton $\CD$ over an~\emph{input $\omega$\=/word} $\word{w}=w_0w_1w_2\cdots\in\albet^\omega$ producing \emph{output $\omega$\=/word} $\word{k}=k_0k_1k_2\cdots\in (\albet_\anIndex)^\omega$
is a~sequence of states $\word{\rho}=\rho_0\rho_1\rho_2\cdots\in (Q_\CD)^\omega$ such that $\rho_0\in\iota_\CD$ and for every $n\in\N$ we have $(\rho_n,w_n,k_n,\rho_{n+1})\in\delta_\CD$.
The $\omega$\=/word $\word{w}$ is \emph{accepted} by $\CD$ if there exists a~run of $\CD$ over $\word{w}$ producing an $\omega$\=/word $\word{k}$ that belongs to $L_\anIndex$.
%Such a~run is \emph{accepting} if the output $\omega$\=/word $k_0k_1k_2\cdots$ belongs to $L_\anIndex$.

The \emph{language} of such an~automaton, denoted $\lang(\CD)\subseteq\albet^\omega$, is the set of $\omega$\=/words $\word{w}\in\albet^\omega$ that are accepted by $\CD$. %such that there exists an~accepting run of $\CD$ over $\word{w}$.
A~language $L\subseteq \albet^\omega$ is \emph{$\omega$\=/regular} if it is the language of some automaton.

An~automaton is \emph{deterministic} if $\iota_\CD$ is a~singleton and the transition relation $\delta_\CD$ is in fact a~function $\delta_\CD\from Q_\CD\times\albet\to\albet_\anIndex\times Q_\CD$, in which case there is a~unique run of $\CD$ over every input $\omega$\=/word $\word{w}\in \albet^\omega$.

\begin{remark}
\label{rem:aut-vs-trans}
If the index $\anIndex=\langle\albet_\anIndex,L_\anIndex\rangle$ is fixed, then deterministic automata $\CD$ over $\albet$ and of index $\anIndex$ are in natural bijection with transducers $\tau\from \albet\tto \albet_\anIndex$ in such a way that $\lang(\CD)=\{\word{w}\in\albet^\omega\mid \tau(\word{w})\in L_\anIndex\}$.
\end{remark}

\subparagraph*{Ramsey theorem.}

Let $C$ be a~finite set of \emph{colours}.
An~\emph{edge labelling} of a~set $X$ is a~function that to each edge $\set{i, j}\subseteq X$ (where $i\neq j$) assigns a~colour from $C$.
Given an~edge labelling, we say that a~set $I\subseteq X$ is \emph{monochromatic} if all edges $\set{i, j}\subseteq I$ have the same colour.

\begin{theorem}[Ramsey]\label{thm:ramsey}
	Let $C$ be a finite set and let $k\in\N$.
	Then, there exists a~computable constant $r\in\N$ such that for every edge labelling of $\set{0,1,\dots,r-1}$ by colours from $C$ there exists a~monochromatic set $I\subseteq\set{0,1,\dots,r-1}$ of size $k$.

	Moreover, for every edge labelling of $\N$ by colours from $C$ there exists an~infinite monochromatic set $I\subseteq\N$.
\end{theorem}

\subparagraph*{Semigroups and monoids.}

An~algebraic structure $\langle S, ({\cdot})\rangle$ with an~associative binary operation $({\cdot})$ is called a~\emph{semigroup}. A~\emph{monoid} is a~semigroup $S$ which contains a~\emph{neutral element} $\epsilon\in S$ such that $\epsilon \cdot s = s\cdot \epsilon = s$ for every $s\in S$. Every semigroup $S$ can be extended into a~monoid $\Peps{S}=S\cup\{\epsilon\}$ by adding a~formal neutral element $\epsilon$ with product defined appropriately. An~\emph{idempotent} is an~element $e\in S$ such that $e\cdot e = e$.

The following fact is a~standard application of Ramsey theorem (cf.\@ \cref{thm:ramsey}).

\begin{fact}\label{thm:ramsey-semi}
	For every finite semigroup $S$ there exists a~computable constant $r\in\N$ such that for every word $s_0s_1\cdots s_{r-1}\in S^r$ there exists a~pair of positions $0\leq i< j<r$
	such that $e\eqdef s_{i+1}\cdot s_{i+2}\cdot \ldots \cdot s_j$ is an~idempotent.

	In particular, putting $c\eqdef s_0\cdot s_1\cdot \ldots \cdot s_j$ we have
	\[c\cdot e = s_0\cdot s_1\cdot \ldots \cdot s_{i-1}\cdot e \cdot e = s_0\cdot s_1\cdot \ldots \cdot s_{i-1}\cdot e = c.\]
\end{fact}

\subparagraph*{Wilke algebras.}

In this work we use Wilke algebras as representations of $\omega$\=/semigroups, as in Perrin and Pin~\cite{perrin_pin_words}.
A~\emph{Wilke algebra} $S$ consists of two sets $(\Wfin{S},\Winf{S})$, two product operations
\begin{align*}
	\Wfin{S}\times \Wfin{S}\to \Wfin{S}
	&&\mbox{and}&&
	\Wfin{S}\times \Winf{S}\to \Winf{S}
\end{align*}
denoted $s\cdot s'$ for operands $s,s'$,
and an operation $\Wfin{S}\to \Winf{S}$ denoted $s^\omega$ for an~operand $s\in\Wfin{S}$. Moreover, the operations are required to satisfy natural associativity axioms, in particular $\Wfin{S}$ needs to be a~semigroup. Each finite Wilke algebra $S$ uniquely determines the \emph{infinite product} operation $\odot\from (\Wfin{S})^\omega\to\Winf{S}$, which is associative. In particular $\odot(sss\cdots) = s^\omega$ and $\odot(s_0s_1\cdots) = s_0\cdot \odot(s_1s_2\cdots)$.

A~\emph{homomorphism} $\alpha$ between two Wilke algebras $S$ and $T$ is a~pair of functions $\Wfin{\alpha}\from \Wfin{S}\to\Wfin{T}$ and $\Winf{\alpha}\from \Winf{S}\to\Winf{T}$ that commute with all the operations of the algebras and with the infinite product $\odot$.

\subparagraph*{Recognition.}

A~canonical example of a~Wilke algebra is $\Wilk{\albet}\eqdef \langle \albet^{+}, \albet^\omega\rangle$, where $\albet$ is an~alphabet. The operations of this Wilke algebra are the concatenation $\cdot$, the infinite repetition $v^\omega\eqdef v\cdot v\cdot v\cdot \ldots\in\albet^\omega$ for $v\in\albet^{+}$, and the infinite product $\odot(v_0v_1v_2\cdots)\eqdef v_0\cdot v_1\cdot v_2\cdot\ldots\in\albet^\omega$ for $v_0,v_1,v_2,\ldots\in\albet^{+}$.

Associativity properties imply that if $\alpha\from \Wilk{\albet}\to S$ is a~homomorphism into a~finite Wilke algebra then for every sequence of finite words $v_0,v_1,\ldots\in\albet^{+}$ we have
\begin{equation}
\label{eq:assoc}
\alpha(v_0\cdot v_1\cdot v_2\cdot \ldots) = \odot\big(\alpha(v_0)\alpha(v_1)\alpha(v_2)\cdots\big).
\end{equation}

Note that if a~Wilke algebra is finite then it can be represented as an~input to an~algorithm by providing its list of elements and ``multiplication tables'' for all the operations. The crucial fact about Wilke algebras is their ability to recognise $\omega$\=/regular languages, as stated by the following theorem.

\begin{theorem}[{\cite{wilke_algebraic}}]
\label{thm:recognition}
Given a~tuple of $\omega$\=/regular languages $(L_0,\ldots,L_{k-1})$ with $L_i\subseteq\albet^\omega$ for all $i<k$, one can effectively compute a~finite Wilke algebra $S$ together with a~homomorphism $\alpha\from \Wilk{\albet}\to S$ and a~tuple of sets $(F_0,\ldots,F_{k-1})$, where for every $i<k$ the set $F_i\subseteq\Winf{S}$ is such that $L_i=\alpha^{-1}(F_i)$. We say that $\alpha$ \emph{recognises} $(L_0,\ldots,L_{k-1})$ with $(F_0,\ldots,F_{k-1})$.

Moreover, one can require $\alpha$ to be \emph{onto} in the sense that $\alpha(\albet^{+})=\Wfin{S}$ and $\alpha(\albet^\omega)=\Winf{S}$.
\end{theorem}

Let $\word{z}=z_0z_1z_2\cdots\in(\Peps{\Wfin{S}})^\omega$
(recall that $\Peps{\Wfin{S}}$ is $\Wfin{S}$ extended with a formal neutral element $\epsilon$).
We say that $\word{z}$ is \emph{saturated} if it contains infinitely many symbols from $\Wfin{S}$, that is, symbols different than $\epsilon$.
In this case $\odot (\word{z})$ is well\=/defined: we can erase all symbols $\epsilon$ from $\word{z}$ obtaining an~$\omega$\=/word $\word{z}'\in (\Wfin{S})^\omega$ and put $\odot(\word{z})\eqdef \odot (\word{z}')$. This definition again satisfies the associativity properties as in \cref{eq:assoc}.

\subparagraph*{Lookahead and composition.}

Assume that $\alpha\from \Wilk{\albet}\to S$ is a~homomorphism into a~finite Wilke algebra $S=(\Wfin{S},\Winf{S})$.
For every $\word{w}=w_0w_1w_2\cdots\in\albet^\omega$ this homomorphism defines the \emph{lookahead} $\lk_\alpha(\word{w})\in (\Winf{S})^\omega$ defined for each position $n\in\N$ as
\[ \big(\lk_\alpha(\word{w})\big)_n \eqdef \alpha\big(w_{n+1}w_{n+2}w_{n+3}\cdots\big)\in\Winf{S}.\]

Note that, while producing a~letter on a~position $n\in\N$, a~transducer uses letters on positions $0,1,\ldots,n$.
On the other hand, a~lookahead at position $n$ depends on positions $n{+}1,n{+}2,n{+}3,\ldots$
To create an~output $\omega$\=/word whose output letters in $\albet_Y$ depend on both the past and the future of input $\omega$\=/words, we consider transducers whose output letters are functions $(\Winf{S}\to \albet_Y)$, and then we apply these functions to letters in $\Winf{S}$ produced by a~lookahead.

To simplify the notation, we use the following shorthand:
if $\word{f}=f_0f_1f_2\cdots\in (\albet_X\to \albet_Y)^\omega$ and $\word{x}=x_0x_1x_2\cdots\in (\albet_X)^\omega$,
then $\funapp{\word{f}}{\word{x}}\in (\albet_Y)^\omega$ is defined for each position $n\in\N$ as $(\funapp{\word{f}}{\word{x}})_n=f_n(x_n)$.

\subparagraph*{Automata over infinite trees.}

A \emph{non\=/deterministic parity tree automaton} over an~alphabet~$\albet$ and of index $\anIndex=\langle{\albet_\anIndex,L_\anIndex}\rangle$
is a~tuple $\CA=\langle\albet,\anIndex,Q_\CA,\iota_\CA,\Delta_\CA\rangle$, where $Q_\CA$ is a~finite set of \emph{states}, $\iota_\CA\subseteq Q_\CA$ a~set of \emph{initial states},
and $\Delta_\CA\subseteq Q\times\albet\times\albet_\anIndex\times Q\times Q$ a~\emph{transition relation}. Again we require the automaton to be complete, that is, for every $q\in Q_\CA$ and $a\in\albet$ it needs to contain at least one transition $(q,a,k,q_\dL,q_\dR)\in\Delta_\CA$.

A \emph{run} of $\CA$ over a~tree $\tree{t}\in\trees_\albet$ producing an~output tree $\tree{\eta}\in\trees_{\albet_\anIndex}$ is a~tree $\tree{\rho}\in\trees_{Q_\CA}$
such that $\tree{\rho}(\epsilon)\in\iota_\CA$ and $(\tree{\rho}(v),\tree{t}(v),\tree{\eta}(v),\tree{\rho}(v\dL),\tree{\rho}(v\dR))\in\Delta_\CA$ for all nodes $v\in\set{\dL,\dR}^\ast$.
A~tree $\tree{t}\in\trees_\albet$ is \emph{accepted} by $\CA$ if there exists a~run of $\CA$ over $\tree{t}$ producing a~tree $\tree{\eta}$ such that for every branch $\word{w}\in\set{\dL,\dR}^\omega$,
the sequence $\tree{\eta}(\word{w}\restr_0)\tree{\eta}(\word{w}\restr_1)\tree{\eta}(\word{w}\restr_2)\cdots\in(\albet_\anIndex)^\omega$ belongs to $L_\anIndex$. The \emph{language} of an~automaton $\CA$ is the set of trees which it accepts. A~language $L\subseteq\trees_\albet$ is a~\emph{regular tree language} if it is the language of some automaton $\CA$.

A~tree automaton is (top\=/down) \emph{deterministic} if $\iota_\CA$ is a~singleton and $\Delta_\CA\from Q\times\albet\to\albet_\anIndex\times Q\times Q$ is a~function.

\subparagraph*{Monadic second\=/order logic.}
\label{pgref:mso}
Formulae of the MSO logic are evaluated in an appropriate structure, which in our case is $\N$ with the successor relation (in the case of $\omega$-words)
or $\set{\dL,\dR}^\ast$ with the left-child and right-child relations (in the case of trees).
Elements of the structure are called positions or nodes.
Usually, a monadic variable in MSO represents a~set of positions, which can be also seen as a word or a tree over the alphabet $\set{0,1}$, with $1$ indicating positions that are in the set.
In this paper, we employ a seemingly more general setting, where each monadic variable $X$ represents a word or a tree over some alphabet $\albet_X$, possibly larger than $\{0,1\}$.
In the sequel, we usually assume a~fixed alphabet $\albet_X$ associated to each variable $X$, but sometimes we explicitly specify the alphabet next to a~quantifier
(writing e.g.,~$\exists X\in(\albet_X)^\omega.\,\varphi(X)$).
Then, for a letter $x\in\albet_X$ and for a~first\=/order variable~$v$ we have an~atomic formula $X(v)=x$ checking whether the letter of $X$ at the position $v$ is $x$.
This way of seeing monadic variables does not increase the expressive power of MSO,
since a~variable with values in $\albet_X$ can be represented by a~tuple of $|\albet_X|$ usual set variables, which should be forced to partition the domain (even $\lceil\log|\albet_X|\rceil$ set variables suffice).

By equivalence between MSO and regular languages~\cite{buchi_decision,mcnaughton_determinisation,rabin_s2s}, we know that for every MSO formula $\varphi(X_1,\dots,X_n)$
we can construct a~deterministic parity $\omega$\=/word automaton (in the case of $\omega$\=/words)
or a~non\=/deterministic parity tree automaton (in the case of trees) over the alphabet $\albet_{X_1}\times\ldots\times\albet_{X_n}$
which accepts exactly those $\omega$\=/words / trees $\word{w}$ over this alphabet for which $\varphi\big(\pi_1(\word{w}),\dots,\pi_n(\word{w})\big)$ holds,
where each $\pi_i(\word{w})$ is obtained from $\word{w}$ by projecting labels of all positions to their $i$\=/th coordinate.
Note that the index of the constructed automaton depends on the formula $\varphi$ and in general cannot be bounded~\cite{bradfield_original,niwinski_nondet_strict,wagner_hierarchy}.

To simplify the notation, we identify a~structure $\word{w}$ over such a~product alphabet $\albet_{X_1}\times\ldots\times\albet_{X_n}$ with the tuple of structures $\langle\pi_1(\word{w}),\dots,\pi_n(\word{w})\rangle$ over respective alphabets.
In particular, for a~formula $\varphi(X_1,\ldots,X_n)$ we can speak about the \emph{language} of a~formula which is defined as the set of structures $\word{w}$ over $\albet_{X_1}\times\ldots\times\albet_{X_n}$ that satisfy $\varphi\big(\pi_1(\word{w}),\dots,\pi_n(\word{w})\big)$. Due to the ability of translating formulae into automata, these languages are always regular.

\subparagraph*{Games.}

We use the general framework of perfect information games of infinite duration played between two players (typically called Player~$\rI$ and Player~$\rII$).
Such a~game is given by a~tuple $\CG=\langle \albet, L_\CG, V_\CG=V^{(\rI)}_\CG\sqcup V^{(\rII)}_\CG, \iota_\CG, \delta_\CG\rangle$ where $\albet$ is an~alphabet,
$L_\CG\subseteq\albet^\omega$ is a~\emph{winning condition}, $V_\CG$ is a~(possibly infinite) set of \emph{positions}, partitioned into the positions of the respective players,
$\iota_\CG\in V_\CG$ is an~\emph{initial position}, and $\delta_\CG\subseteq V_\CG\times \albet\times V_\CG$ is an~\emph{edge relation}
(again satisfying completeness property that each $v\in V_\CG$ admits at least one edge $(v,a,v')\in\delta_\CG$).
The letter $a\in\albet$ is called the \emph{label} of an~edge $(v,a,v')\in\delta_\CG$.

A~play of such a~game is played in rounds, with the initial position $v_0=\iota_\CG$.
In round number $n\in\N$ the player $P$ such that $v_n\in V^{(P)}_\CG$ chooses an~edge $(v_n,k_n,v_{n+1})\in\delta_\CG$ moving to the next position $v_{n+1}$.
After an~infinite play, Player~$\rII$ wins if and only if $\word{k}\eqdef k_0k_1k_2\cdots$ belongs to $L_\CG$.
Classical theorems~\cite{martin_borel_determinacy} imply that if $L_\CG$ is sufficiently simple, then one of the players can ensure to win this game, that is, has a~\emph{winning strategy}.
In general such a~strategy for a~player $P$ is a~tree\=/shaped object but we mostly work with \emph{positional strategies}, that is,
functions $\sigma^{(P)}\from V^{(P)}_\CG\to \delta_\CG$ such that for every $v\in V^{(P)}_\CG$ we have $\sigma^{(P)}(v)=(v,k,v')$ for some $k\in\albet$ and $v'\in V_\CG$.

A~\emph{parity game} of index $\anIndex=\langle \albet_\anIndex, L_\anIndex\rangle$ is a~game $\CG$ as above where $\albet=\albet_\anIndex$ and $L_\CG=L_\anIndex$.

\begin{theorem}[{\cite{jutla_determinacy,mostowski_parity_games}}]
\label{thm:parity-determined}
If $\CG$ is a~parity game then some player $P$ has a~positional winning strategy $\sigma^{(P)}$ in $\CG$.
\end{theorem}

\section{New quantifiers}

In this section we introduce the two types of quantifiers which are studied in this work.
When doing so, we follow the convention to assume that in a~formula $\qQ X.\,\varphi(W_1,\ldots,W_k,X)$
all the parameter variables $W_1,\ldots,W_k$ are combined into a~single free variable $W$ over a~product alphabet, as explained above.
Thus, we focus on formulae of the form $\qQ X.\,\varphi(W,X)$, even if the respective coordinates of $W$ come from different outer quantifiers.

\subparagraph*{Index quantifiers.}

Consider a~new quantifier $\qI^{D}_\anIndex X.\,\varphi(W,X)$ where $D\in\{\Idet, \Indet\}$ determines the type of involved automata and $\anIndex$ is an~index (either a~strong parity index $\CP_{i,j}$ or a~weak parity index $\CW_{i,j}$).
Such a~formula holds for a~parameter $\word{w}\in(\albet_W)^\omega$ if there exists an~automaton~$\CA_{\word{w}}$ of index $\anIndex$, which is either deterministic ($D=\Idet$) or non\=/deterministic ($D=\Indet$),
such that for every $\word{x}\in(\albet_X)^\omega$ the formula $\varphi(\word{w},\word{x})$ holds if and only if $\CA_{\word{w}}$~accepts $\langle\word{w},\word{x}\rangle$.

Note that the parameter $\word{w}$ occurs in the above definition in two roles.
First, the automaton~$\CA_{\word{w}}$ may depend on the parameter $\word{w}$.
Second, the automaton, when verifying whether the given $\word{x}$ makes $\varphi(\word{w},\word{x})$ true, has access not only to $\word{x}$ but also to the parameter $\word{w}$
(in particular, the automaton is over the alphabet $\albet_{W}\times \albet_X$).

\begin{remark}
\label{rem:index-no-access}
One may ask what changes if we consider another semantics of the index quantifier, where the hypothetical automaton does not have access to the parameters $\word{w}$ but only reads the quantified $\omega$\=/word $\word{x}$. In this case the formalism becomes immediately undecidable. Indeed, consider the simplest possible formula $\qI^{\Idet}_{\CW_{0,1}} X.\,(X=W)$, which involves the deterministic safety index quantifier. Then, for a~given~$\word{w}$ the set of $\omega$\=/words $\word{x}$ that satisfy $\word{x}=\word{w}$ is $\{\word{w}\}\subseteq(\albet_W)^\omega$. This language is recognised by a~deterministic safety automaton if and only if $\word{w}$ is ultimately periodic.
Due to Bojańczyk et al.~\cite{bojanczyk_undecidability}, this extended logic is undecidable.
\end{remark}

Note that both deterministic and non\=/deterministic index quantifiers make sense for both $\omega$\=/words and trees. Let MSO+$\qI$ denote the extension of monadic second\=/order logic by index quantifiers.

\subparagraph*{Game quantifiers.}

As a~natural way to study the index quantifier, we need to formalise within MSO the concept of the \emph{game quantifier} $\qG$ (see~\cite[\S~20.D]{kechris_descriptive} and~\cite{bradfield03,bradfield_transfinite,finkel_upper_tree,kanovei_survey,damian_henryk,moschovakis_inductive}).
This quantifier, written $\qG X\sto Y$ (alternatively, in some papers the symbol $\Game$ is used), binds two monadic variables $X$ and $Y$.
A formula
\[\qG X\sto Y.\, \varphi(W,X,Y)\]
holds, given a~parameter $\word{w}\in(\albet_W)^\omega$, if Player~$\rII$ has a~winning strategy in the game $\theGame(\word{w},\varphi)$, defined as follows.\footnote{%
Classically, in the works of Moschovakis and Kechris~\cite{kechris_descriptive,moschovakis_inductive} the ``game quantifier'' requires Player~$\rI$ to win the game,
however in automata\=/theoretic context (e.g.,~the Church synthesis problem~\cite{buchi_synthesis,rabin_church_trees}) or Wadge games~\cite{wadge_phd},
it is more customary to focus on Player~$\rII$.}
The game consists of infinitely many rounds.
In a~round $n\in\N$, Player~$\rI$ proposes a~letter $x_n\in\albet_X$ and Player~$\rII$ answers with a~letter $y_n\in\albet_Y$.
At the end, Player~$\rII$ wins if and only if $\varphi(\word{w},\word{x},\word{y})$
holds for $\word{x}\eqdef x_0x_1x_2\cdots\in(\albet_X)^\omega$ and $\word{y}\eqdef y_0y_1y_2\cdots\in(\albet_Y)^\omega$.
This game can easily be represented by a~formal game $\CG$ with positions $V^{(\rI)}_\CG\eqdef \N$ and $V^{(\rII)}_\CG\eqdef\N\times \albet_X$ and $L_\CG$ given by $\varphi$; however we do not need to study the exact structure of this game.

Typically, one applies the game quantifiers in the context where the involved games are determined, although the definition makes sense even without this assumption.

Let MSO+$\qG$ denote the extension of monadic second\=/order logic by game quantifiers.
Note that as it is defined, the game quantifier makes sense only for $\omega$\=/words, because the shape of the time\=/structure of a~game of infinite duration is $\omega$.

\section{Generalised quantifiers}
\label{sec:generalised}

In this section we relate the quantifiers introduced in this paper to the general concept of \emph{generalised quantifiers}. They were proposed by Mostowski~\cite{mostowski-quantifiers} as
an~abstract logical construct that generalises the classical quantifiers $\exists$ and $\forall$. Since then, they became an~important tool in various applications of logic (see, e.g.,~\cite{sep-generalized-quantifiers} for a~survey).

At the syntactic level, a~quantifier $\qQ$ extends the language by a~construction $\qQ x.\,\varphi (\vec{w},x)$, for an~arbitrary formula $\varphi$. Here, a~variable $x$ is \emph{bound} by $\qQ$, whereas the variables in $\vec{w} = (w_1, \ldots ,w_k)$ remain \emph{free}.
At the semantic level, the quantifier is associated with an~operator, which, for any structure $\CM$ (with universe $M$) defines a~family of sets $\qQ^{\CM} \subseteq \powerset(M)$.
Then, given a~valuation $\vec{w} \mapsto \vec{a}\in M^k$, the formula $\qQ x.\,\varphi (\vec{a},x)$ holds in $\CM$ if the set $\{ b\in M \mid \text{$\varphi (\vec{a},b)$ holds in $\CM$}\}$ belongs to $\qQ^{\CM}$.
In this setting, $\exists^{\CM}$ is the family of all non\=/empty subsets of $M$, whereas $\forall^{\CM} = \{ M \}$.
As a~less standard example, one can express the property that the cardinality of the set of $x$'s satisfying $\varphi (\vec{w},x)$ belongs to some specified class of cardinals (i.e.,~$\exists^\infty$ says that the set is infinite), or that the set of $x$'s that do satisfy~$\varphi$ and those that do not, have the same cardinality. It is usually assumed that the family $\qQ^{\CM}$ is invariant under permutations of $M$, but a weakening of this requirement is sometimes justified.

More generally, one can consider $n$\=/ary quantifiers, where a~quantifier $\qQ$ bounds simultaneously $n$ variables and,
respectively, $\qQ^{\CM}$ is a~family of $n$\=/ary relations over $M$.
For example, if $n=2$ and $\qQ^{\CM}$ is the class of rectangles, that is, $\qQ^{\CM} = \{ X \times Y \mid X,Y \in \powerset(M) \}$ then $\qQ x y.\,\varphi (\vec{w},x,y)$ expresses the fact that
whenever $\varphi (\vec{a},b_1,c_1)$ and $\varphi (\vec{a},b_2,c_2)$ hold in $\CM$ then $\varphi(\vec{a},b_1,c_2)$ and $\varphi(\vec{a},b_2,c_1)$ hold as well.

One can adapt the above concepts to monadic second\=/order logic (MSO), with
$\qQ^{\CM} \subseteq \powerset(\powerset (M))$ in the unary case, and in general
$\qQ^{\CM} \subseteq \powerset\left(\left(\powerset (M) \right)^n \right)$. Indeed, several generalised quantifiers of this kind have been considered in the literature, the eminent example being the \emph{weak} quantifiers, that is, the quantifiers $\exists$ and $\forall$ restricted to \emph{finite} sets.
The cardinality quantifiers and unboundedness quantifiers mentioned in the introduction can also be presented in this framework.
%\ignore{Boja\'nczyk~\cite{bojanczyk_bounding} considered an~\emph{unboundedness} quantifier ${\mathbb U} X$, expressing the property that a formula is satisfied by finite sets $X$ of unbounded size.
%B\'ar\'any, Kaiser, and Rabinovich~\cite{barany_expressing_trees} gave attention to the \emph{cardinality quantifiers} in MSO logic over the binary tree, like, e.g.,~$\exists^{\geq \kappa} X.\,\varphi(\vec{Y}, X)$, stating that there are at least $\kappa$ distinct sets $X$ satisfying $\varphi(\vec{Y}, X)$. The definition of $\qQ^{\CM}$ in the above examples is straightforward.
%The expressive power varies:
%%The aforementioned authors showed that
%the cardinality quantifiers can be in fact expressed in the standard syntax of MSO~\cite{barany_expressing_trees}, thus the extended theory remains decidable.
%But the unboundedness quantifier cannot, and the extension becomes undecidable~\cite{bojanczyk_msou_final}.
%However,
%Niwiński, Parys, and Skrzypczak exhibit an~example of a~quantifier related to the ordinal rank of well\=/founded sets that cannot be expressed in MSO, although the respective property expressed by this quantifier is decidable.
%\pp{Także kwantyfikatory $U$ (oraz Nabla). Mają nieco inny charakter.}
%\dn{Napisałem coś o $U$ i o {\em game}. Nabla mniej mi tu pasują, może nie trzeba wszystkiego (?). Wzmianka o naszym kwantyfikatorze {\em ordinal rank\/} też do rozważenia.}}

\subparagraph*{Game quantifiers.}
We begin by discussing how game quantifiers introduced above can be viewed as generalised MSO quantifiers over the structure $\N$. To explain the idea, let us first take a~simple example in first\=/order logic. Consider a~formula
\[ \forall x.\, \exists y.\, \forall x'.\, \exists y'.\, \varphi (\vec{w},x,y,x',y'). \]
Clearly, its meaning in a structure $\CM$ can be viewed as a game of two players, say ${\bf \exists}$ and~${\bf \forall}$, consisting of $4$ rounds. Now the block of $4$ quantifiers can be replaced by a~single $4$\=/ary quantifier, so that the formula becomes
$\qQ xyx'y'. \, \varphi (\vec{w},x,y,x',y')$. The semantics of $\qQ$ is specified by a property that a $4$\=/ary relation $r$ in $\qQ^{\CM}$ should possess. In terms of a~game, in which Players~$\rI$ and~$\rII$ select in alternation elements of $M$, Player~$\rII$ should have a strategy to force the selected quadruple into $r$.
%\dn{Maybe the fragment about first-order games should be removed for lack of space.}

Now consider a formula $\varphi (W,X,Y)$ interpreted in the structure $\N$, where $W$, $X$, $Y$ are set variables (more generally, they could be some tuples of set variables). Consider an infinite game, in which Players~$\rI$ and~$\rII$ select in alternation bits in $\{ 0,1 \}$, so that the result is an~infinite sequence
\[
x_0, y_0, x_1, y_1, x_2, y_2, \ldots, x_n, y_n, \ldots
\]
The sequences $x_0, x_1, x_2, \ldots$ and $y_0, y_1, y_2, \ldots$ constitute characteristic functions of some subsets $\word{x}$ and $\word{y}$ of $\N$, respectively.
Now, for a~valuation $W \mapsto \word{w}$, a~formula defined with the game quantifier
\[
\qG X \sto Y.\,\varphi (\word{w},X,Y)
\]
holds if Player~$\rII$ has a strategy to force that the formula $\varphi (\word{w}, \word{x}, \word{y})$ holds in $\N$.
The game quantifier $\qG$ can be defined as a {\em binary\/} generalised MSO quantifier.
Its semantics $\qG^{\N}$ is defined by a family of binary relations over $\powerset(\N)$ that comprises \emph{all} relations $R \subseteq \powerset (\N) \times \powerset (\N)$,
such that in the game described above, Player~$\rII$ has a~strategy to force the resulting pair $(\word{x}, \word{y})$ into $R$.
Then, indeed, the formula $\qG X \sto Y .\, \varphi (\word{w},X,Y)$ holds precisely when the relation $\{ (\word{x}, \word{y}) \mid \text{$\varphi (\word{w}, \word{x}, \word{y})$ holds in $\N$} \}$ belongs to $\qG^{\N}$.

\subparagraph*{Index quantifiers.}

To present our new index quantifier $\qI^{D}_\anIndex X.\,\phi(\vec{W},X)$ as a generalised MSO quantifier, let us,
for concreteness, focus on the MSO theory of
the full binary tree, whose domain is $\{\dL,\dR\}^\ast$.
%\ignore{One may be tempted to introduce new quantifiers
%expressing those properties of sets (or tuples thereof) that are
%related to automata, where we identify a~tuple of sets $(X_1, \ldots,
%X_n)$ with its characteristic function viewed as a~labelled tree over
%alphabet $\{0,1\}^n$. For example, one may consider a~quantifier
%$\exists^B X$, expressing the property that the set of $X$'s satisfying the formula is recognisable by a~Büchi tree automaton. There is however, an~evidence that the expressive power of such a~quantifier
%would be limited. Indeed, consider a formula $\varphi (X,Y)$ and
%let, for $W \subseteq T$, $\varphi^{-1} (W)$ denote the set of $V$'s
%such that $\varphi (V,W)$ is true. Suppose that
%$\exists^B X.\,\varphi (X,{W})$ holds for some
%``wild'' (in particular, non\=/regular) set ${W} \subseteq T$.
%Then by the Regular Tree Theorem there must be a \emph{regular} set
%$W_0 \subseteq T$, such that
%Let $\varphi^{-1} (W) = \varphi^{-1} (W_0)$. Hence, in general, either the formula $\exists^B X.\,\varphi (X,{W})$ is ``often'' false, or the formula $\varphi (X,Y)$ only weakly ``correlates'' its arguments.
%
%\dn{The above remark is maybe too complicated at this
%place; it can be safely omitted. But I leave it for the moment, as
%it points out to an interesting general question: how strongly a
%relation $\varphi (X,Y)$ correlates its arguments.}
%
%Our definition of the quantifier $\qI^D_{\alpha} X.\,\varphi(X,\vec{Y})$ introduced in \cref{sect-index} above follows different lines. We have assumed that the automaton in consideration reads \emph{both} $X$ and $\vec{Y}$.}
As we have assumed that our automaton reads the values of both $\vec{W}$ and $X$,
the construction does not fit into
the unary case, but, like the game quantifier, it can be expressed as a~\emph{binary}
quantifier, or more generally, $(k{+}\ell)$\=/ary quantifier (if $\vec{W}$ is a~$k$\=/vector and $X$ an~$\ell$\=/vector).

For simplicity, let us consider $k=\ell=1$; an~extension to higher $k$, $\ell$ is straightforward. The key
point is to choose a~class of binary relations over $\powerset(\{\dL,\dR\}^\ast)$ that would serve as the intended semantics of the quantifier.
For a~binary relation $r \subseteq \powerset(\{\dL,\dR\}^\ast) \times \powerset(\{\dL,\dR\}^\ast)$, and a~set $K \in \powerset (\{\dL,\dR\}^\ast)$, we define the \emph{cut} of $r$ by $K$ as the binary relation
\[
  r_K = r \, \cap \, \big(\{ K \} \times \powerset (\{\dL,\dR\}^\ast) \big)
      = \big\{ (K,L)\in r \mid L \in \powerset(\{\dL,\dR\}^\ast) \big\}.
\]
Recall that in our quantifier we are interested in automata of type $D$ and index
$\anIndex$. A pair of sets $(K,L)$ is accepted by an~automaton
(over the alphabet $\{0,1\}^2$) if so is its characteristic function, and
a~relation $r \subseteq \powerset(\{\dL,\dR\}^\ast) \times \powerset(\{\dL,\dR\}^\ast)$ is recognised by an~automaton if it consists precisely of pairs that the automaton
accepts.
Now consider the class of relations
\[ \CC^D_{\anIndex } = \big\{ r_K \mid K \in \powerset(\{\dL,\dR\}^\ast) \; \land \; \text{$r$ is recognised by an~automaton of type $D$ and index $\anIndex$}\big\}.\]
Then it is straightforward to see that the formula $\qI^D_{\anIndex} X.\,\varphi(W,X)$ is equivalent to
\[ \qQ^D_\anIndex ZX.\, \varphi (Z,X) \, \land \, Z=W, \]
where the semantics of the quantifier $\qQ^D_\anIndex$ over trees is given by the class $\CC^D_{\anIndex}$.

Clearly, the variable $Z$ above plays only a~technical role; therefore,
for clarity of notation, in our paper we use the notation $\qQ^D_\anIndex X.\, \varphi (W,X)$, without $Z$.

Let us also remark that our proposal is not the only possible approach. One could also consider a~unary quantifier $\qQ^D_\anIndex X$, where a~formula
$\qQ^D_\anIndex X.\, \varphi (W,X)$ holds for a valuation $W \mapsto \tree{w}$ if
the language of all sets $\tree{x}$ such that $\varphi (\tree{w}, \tree{x})$ holds is accepted by an~automaton (of appropriate kind), without reading the parameter $\tree{w}$, as discussed in \cref{rem:index-no-access}.
That is, the semantics is given simply by a~class of all languages accepted by automata of type $D$ and index $\anIndex$.

While this may appear quite natural, we believe that such an extension would be less interesting.
Not only it brings an~undecidable formalism over $\omega$\=/words as indicated in \cref{rem:index-no-access} but it additionally restricts available correlation between the involved variables.
Indeed, if such a~formula is satisfied by some $\tree{w}$ which is not regular, then it follows from general properties of MSO (namely Regular Tree Theorem) that there is a~regular $\tree{w}'$, such that the languages $\{ \tree{x} \mid \varphi(\tree{w}, \tree{x})\}$ and $\{ \tree{x} \mid \varphi(\tree{w}', \tree{x})\}$ coincide.
Thus the relation defined by the formula $\varphi (W, X)$, in some sense, necessarily weakly correlates its arguments.
These issues require further investigation.

%\ignore{The purpose of the current
%subsection was to show that our construction fits into the general
%framework, which justifies the use of the term \emph{quantifier}.
%
%
%We have shown how the quantifier of \cref{sect-index}
%can be incorporated into a general framework. However, these
%formulas form only a special case of what should be normally written
%in this framework, as they oblige the automaton to read {\em all }
%variables $\vec{Y}$. I believe that a more appropriate form would be
%\[
%\qI_{\alpha} X | Y.\,\varphi(X,\vec{Y}, \vec{Z}),
%\]
%where the formula holds for a~vector of variables $\vec{Y}, \vec{Z}$ if
%there exists an~automaton~$\CA$ of index $\alpha$ such that for every
%$X$ the formula $\varphi(X,\vec{Y}, \vec{Z})$ holds if and only if
%$\CA$~accepts $(X,\vec{Y})$. Here $\vec{Y}$ plays the role of a
%{\em condition\/}, with a remote analogy to conditional probability,
%or conditional entropy:}

\section{Game quantifiers over \texorpdfstring{$\omega$\=/}{w-}words}

The first part of our results concerns the~\emph{game quantifier} $\qG$. We start by showing that the extended formalism of MSO+$\qG$ can be reduced back to pure MSO, that is, the game quantifiers can be eliminated.
However, our goal is to obtain a stronger property, stated in \cref{thm:transducer}:
under appropriate assumptions on the formula, games described by quantifiers~$\qG$ admit strategies that can be realised by finite-memory transducers.

Consider an~instance of a~game quantifier $\qG X\sto Y.\,\varphi(W,X,Y)$, where the internal formula $\varphi(W,X,Y)$ is in MSO.

\begin{lemma}[Folklore]\label{lem:elim-1-game-quant}
	For every formula of the form $\qG X\sto Y.\,\varphi(W,X,Y)$, where $\varphi$ is in MSO,
	one can effectively construct an~equivalent formula of pure MSO.
\end{lemma}

This construction can be found in a work by Kaiser~\cite{kaiser_game_quant}.
We include a proof for the sake of completeness.
The concepts introduced in this proof will be useful later on in the paper.

\begin{proof}
Let $\CD$ be a~deterministic parity automaton over the alphabet $\albet_W\times\albet_X\times\albet_Y$ of a~strong parity index $\anIndex=\langle \{i,i{+}1,\ldots,j\},P_{i,j}\rangle$ that is equivalent to $\varphi$,
that is, the automaton accepts an~$\omega$\=/word $\langle\word{w},\word{x},\word{y}\rangle$ if and only if $\varphi(\word{w},\word{x},\word{y})$ holds.

Given an~$\omega$\=/word $\word{w}=w_0w_1w_2\dots\in(\albet_W)^\omega$,
we can consider a~parity game $\theGame(\word{w},\CD)$ obtained as a~product of $\theGame(\word{w},\varphi)$ with the automaton $\CD$, defined as follows.

\begin{definition}
\label{def:product-game}
Let $Q$ be the set of states of $\CD$, and $\delta$ its transition function.
The set of positions of $\theGame(\word{w},\CD)$ is then given by $V^{(\rI)}\eqdef \N\times Q$ and $V^{(\rII)}\eqdef \N\times Q\times \albet_X$.
From a~position $(n,q)\in\N\times Q$ first Player~$\rI$ proposes $x_n\in \albet_X$ and the game moves to the position~$(n,q,x_n)$. Then Player~$\rII$ proposes $y_n\in\albet_Y$ and the game moves to the position~$(n{+}1,q')$ where $\delta(q, (w_n, x_n, y_n))=(k_n, q')$. The label of the former edge equals the lowest priority $i$ (i.e.,~is irrelevant), while the label of the latter edge equals $k_n$.
\end{definition}

It is easy to see that $\theGame(\word{w},\CD)$ is equivalent to $\theGame(\word{w},\varphi)$ in the sense that a~player $P$ wins one game if and only if she wins another:
the automaton $\CD$ is deterministic, so there is a one\=/to\=/one correspondence between choices in $\theGame(\word{w},\varphi)$ and choices in $\theGame(\word{w},\CD)$,
so that strategies from one game can be directly transferred to the other game.
Moreover, due to positional determinacy of parity games (see \cref{thm:parity-determined}),
Player~$\rII$ wins $\theGame(\word{w},\varphi)$ if and only if Player~$\rII$ has a~positional winning strategy in $\theGame(\word{w},\CD)$.

A~positional strategy $\sigma^{(\rII)}$ of Player~$\rII$ in $\theGame(\word{w},\CD)$ can be represented by an~$\omega$\=/word $\word{\sigma}=\sigma_0\sigma_1\sigma_2\cdots\in \big(Q\times\albet_X\to\albet_Y\big)^\omega$, where $Q$ is the set of states of $\CD$: in this $\omega$\=/word, the letter $\sigma_n$ satisfies $\sigma_n(q,x)=y$ where $\sigma^{(\rII)}(n,q,x)=\big((n,q,x),k,(n{+}1,q')\big)$ with $\delta_\CD\big(q,(w_n,x,y))=(k,q')$.
The following claim is straightforward, as MSO allows us to quantify over infinite plays in $\theGame(\word{w},\CD)$ and can express the parity condition $\CP_{i,j}$.

\begin{claim}\label{claim:formula-for-strategy}
There exists an~MSO formula $\psi^{(\rII)}(W,\Sigma)$
such that $\psi^{(\rII)}(\word{w},\word{\sigma})$ holds for $\word{w}\in(\albet_W)^\omega$ and $\word{\sigma}\in \big(Q\times\albet_X\to\albet_Y\big)^\omega$ if and only if
$\word{\sigma}$ encodes a~positional winning strategy~$\sigma^{(\rII)}$ of Player~$\rII$ in $\theGame(\word{w},\CD)$.
\end{claim}

It follows that the formula $\qG X\sto Y.\,\varphi(W,X,Y)$ is equivalent to
\[\exists \Sigma\in\big(Q\times\albet_X\to\albet_Y\big)^\omega.\,\psi^{(\rII)}(W,\Sigma),\]
where the set $\big(Q\times\albet_X\to\albet_Y\big)$ is finite and therefore one can treat it as an~alphabet. Consequently, this formula belongs to pure MSO.
\end{proof}

Using the above lemma to inductively eliminate an~innermost game quantifier, we immediately obtain the following corollary.

\begin{corollary}
\label{cor:elim-game-quant}
	The expressive power of MSO+$\qG$ is equal to that of MSO. Moreover, there exists an~effective procedure that eliminates the game quantifiers.
\end{corollary}

\begin{remark}
\label{rem:effective-strategy}
	Consider formulae without the parameter $W$, that is, $\qG X\sto Y.\,\varphi(X,Y)$, where $\varphi$ is in MSO. In this case the game $\theGame(\word{w},\CD)$ can be played over the arena $Q$ instead of $\N\times Q$.
	Thus, it is a~finite parity game, which can be solved directly.
	The resulting strategy $\sigma^{(\rII)}$ takes the shape of a~transducer $\tau\from \albet_X\tto\albet_Y$ (its set of states is just~$Q$) such that for every $\word{x}\in(\albet_X)^\omega$ we have $\varphi\big(\word{x},\tau(\word{x})\big)$.
\end{remark}

The above remark can be seen as a~modern version of a~proof of the B\"uchi-Landweber theorem~\cite{buchi_synthesis}, based on determinacy of parity games.
This means that the proposed procedure of elimination of a~game quantifier can be seen as a~parametrised version of the construction of B\"uchi and Landweber,
where we search for a~strategy that may depend on the parameter $\word{w}\in(\albet_W)^\omega$.

\subsection{Sequential strategies}

One may ask if it is possible to recover some version of \cref{rem:effective-strategy} in the presence of external parameters $W$, namely represent the strategy $\sigma^{(\rII)}$ as a~transducer.
Of course the exact strategy may depend on the global properties of $W$,
so one cannot expect to have a~single transducer $\tau\from \albet_W\times\albet_X\tto\albet_Y$ that would realise the strategy.
However, what happens if we allow the transducer to depend on a~given $\omega$\=/word $\word{w}$?

\begin{question}
	Assume that for some parameter $\word{w}\in(\albet_W)^\omega$ a~formula $\qG X\sto Y.\,\varphi(\word{w},X,Y)$ holds.
	Does it mean that there exists a~transducer $\tau_{\word{w}}\from\albet_W\times\albet_X\tto\albet_Y$ that realises a~winning strategy of Player~$\rII$ in $\theGame(\word{w},\varphi)$? In other words, we ask if we can ensure that
	\begin{equation}
	\label{eq:transducer}
	\text{for every $\word{x}\in(\albet_X)^\omega$ we have $\varphi\big(\word{w},\word{x},\tau_{\word{w}}(\word{w},\word{x})\big)$.}
	\end{equation}
\end{question}

It turns out that the answer is negative---the strategies used by Player~$\rII$ may not be made finite\=/memory, even if $\word{w}$ is known in advance.
Intuitively, this boils down to the fact that $\word{w}$ may not be ultimately periodic, while $\varphi$ may require some position\=/to\=/position correspondence between $Y$ and $W$.
More precisely, we have the following fact.

\begin{restatable}{fact}{factNoTransducer}
\label{ft:no-transducer}
	There exists a formula $\varphi(W,X,Y)$ in MSO such that for some concrete $\omega$\=/word $\word{w}\in(\albet_W)^\omega$ we have $\qG X\sto Y.\,\varphi(\word{w},X,Y)$
	while no transducer $\tau\from (\albet_W\times\albet_X)\tto\albet_Y$ satisfies \cref{eq:transducer}.
\end{restatable}

\begin{proof}
	Let $\albet_W=\albet_Y=\{0,1\}$ and let $\varphi(\word{w},\word{x},\word{y})$ for $\word{w}=w_0w_1w_2\cdots$ and $\word{y}=y_0y_1y_2\cdots$ say that for every $n\in\N$ we have $y_n=w_{n+1}$.
	Notice that for all $\word{w}\in(\albet_W)^\omega$ we have $\qG X\sto Y.\,\varphi(\word{w},X,Y)$ because $X$ plays no role in $\varphi$
	and it is enough for Player~$\rII$ to play consecutive values $y_0\eqdef w_1$, $y_1\eqdef w_2$, and so on.

	Let $\word{w}$ be defined as $0^1 1 0^2 1 0^3 1 \cdots\in (\albet_W)^\omega$.
	It remains to show that no transducer $\tau_{\word{w}}\from (\albet_W\times\albet_X)\tto\albet_Y$ satisfies \cref{eq:transducer}.
	Assume to the contrary that $\tau_{\word{w}}$ is such a~transducer with a~set of states $Q$ and a~transition function $\delta$.
	Fix any letter $x\in\albet_X$, and consider the unique run of $\tau$ over the $\omega$\=/word $\word{w}$ defined above and over $\word{x}=xxx\cdots\in(\albet_X)^\omega$.
	Take any $n\geq|Q|$, and concentrate on the fragment of this run reading the infix $10^{n+1}1$ of~$\word{w}$.
	The transducer should produce $0$'s while reading the first $n$ zeroes of the input fragment (because the next input letter is $0$),
	and $1$ over the last zero (because the next input letter is $1$).
	Let $\rho_0,\rho_1,\ldots,\rho_{n+1}$ be the states of $\tau$ visited over this fragment, with $\rho_0$ before the first $0$, and $\rho_{n+1}$ after the last $0$.
	By the pigeonhole principle, we have $\rho_k=\rho_\ell$ for some $k$, $\ell$ with $0\leq k<\ell\leq n$.
	For $i\in\set{0,\dots,n{-}1}$ we have $\delta(\rho_i,(0,x))=(0,\rho_{i+1})$, which applied to consecutive positions after $k$ and $\ell$ implies $\rho_{k'}=\rho_n$, where $k'=k+(n-\ell)<n$.
	But then $(0,\rho_{k'+1})=\delta(\rho_{k'},(0,x))=\delta(\rho_n,(0,x))=(1,\rho_{n+1})$.
	In other words, the transducer has no way of counting where to produce a~$1$, if the number of zeroes exceeds the number of its states.
\end{proof}

\ignore{
The idea is to take $\varphi(W,X,Y)$ which does not depend on $X$ and requires that $y_n=w_{n+1}$ for every $n\in\N$. Such a~formula in a~sense forces the strategy to use some kind of \emph{lookahead}. Then clearly $\qG X\sto Y.\,\varphi(\word{w},X,Y)$ holds, but for $\word{w}$ which is not ultimately periodic no transducer can realise the respective strategy. A~complete proof is given in \cref{app:no-transducer}.
}

This negative answer can be explained from two perspectives. One, directly suggested by the above example, focuses on the need of a~lookahead---if $\tau_{\word{w}}$ was able to perform some lookahead to the future of the parameter word ${\word{w}}$, then it could easily realise the respective strategy. This observation is formalised in \cref{lem:allow-lookahead}, where the lookahead is allowed. This approach follows similar lines as the results of Winter and Zimmermann~\cite{winter-delay-games}, where the authors study games with lookahead.

Another point of view is that in contrast to the construction by B\"uchi and Landweber~\cite{buchi_synthesis} (see also~\cite{thomas-solving}), the arena of $\Game(\word{w},\CD)$ is infinite. Thus, some subtle synchronisation between the variables may go on indefinitely. To avoid this problem, we consider the notion of a~formula that \emph{depends separately} on one variable (see \cref{ssec:separate}). It turns out that in this case the infiniteness of the arena stops being a~problem and the strategies can again be realised by transducers, as stated in \cref{thm:transducer}.

\subsection{Uniformisation by transducers with lookahead}
\label{ssec:uniformisation}

Before we move on, we need to first show how \emph{uniformised} relations can be realised by transducers with lookahead. We say that an~MSO formula $\psi(W,Y)$ is \emph{uniformised} if for every~$\word{w}$ there exists at most one $\word{y}$ such that $\psi(\word{w},\word{y})$ holds.
The next fact states that a~partial function described by a~uniformised MSO\=/formula can be realised by a~transducer composed with a~lookahead.
This fact is rather general and almost folklore; it relies on the composition method for MSO~\cite{shelah_composition} (expressed by Wilke algebras in our setup).

\begin{restatable}{fact}{ftUniformisedLookahead}
\label{ft:uniformised-lookahead}
	Assume that $\psi(W,Y)$ is uniformised.
	Then, one can effectively construct a~homomorphism $\alpha\from \Wilk{(\albet_W)}\to S$ onto a~finite Wilke algebra $S$
	together with a~transducer $\tau\from \albet_W\tto (\Winf{S}\to \albet_Y)$
	such that for every $\word{w}\in(\albet_W)^\omega$ for which $\exists Y.\,\psi(\word{w},Y)$ holds we have
	\[ \psi\big(\word{w}, \funapp{\tau(\word{w})}{\lk_\alpha(\word{w})}\big).\]
\end{restatable}

In other words, for an~input $\omega$\=/word $\word{w}=w_0w_1w_2\cdots$ and $f_0f_1f_2\cdots\eqdef \tau(\word{w})$ we consider $h_n\eqdef \alpha(w_{n+1}w_{n+2}\cdots)$ and $y_n\eqdef f_n(h_n)$ defined for $n=0,1,2,\ldots$, and claim that $\psi(\word{w},y_0y_1y_2\cdots)$ holds.

\begin{proof}
	For $n\in\N$ let $\chi_n\in\set{0,1}^\omega$ denote the $\omega$\=/word having $1$ at the position $n$, and zeroes everywhere else.
	For each $y\in\albet_Y$ consider a~formula $\psi_y(W,Z)$ such that, assuming $\exists Y.\,\psi(\word{w},Y)$, we have $\psi_y(\word{w},\chi_n)$
	if the letter at position $n$ of the unique $\word{y}$ such that $\psi(\word{w},\word{y})$ holds equals $y$;
	such a formula can be easily constructed out of $\psi$.

	Apply \cref{thm:recognition} to the tuple of languages defined by formulae $\big(\psi_y(W,Z)\big)_{y\in\albet_Y}$ to obtain a~homomorphism $\beta\from\Wilk{(\albet_W\times\{0,1\})}\to S$ onto a~finite Wilke algebra $S$ together with a~tuple of sets $(F_y)_{y\in \albet_Y}$ such that $\psi_y(\word{w},\word{z})$ holds if and only if $\beta\big(\langle\word{w},\word{z}\rangle\big)\in F_y$.

	Let $\add_0\from\Wilk{(\albet_W)}\to \Wilk{(\albet_W\times\{0,1\})}$
	be the homomorphism adding $0$ on the second coordinate of all letters in a~given word.
	Then as $\alpha\from \Wilk{(\albet_W)}\to S$ we take $\add_0\circ\beta$.

	Next, we construct the transducer $\tau\from \albet_W\tto (\Winf{S}\to \albet_Y)$.
	It remembers the value under~$\alpha$ of the prefix read so far.
	To this end, its set of states is the monoid $\Peps{\Wfin{S}}$ obtained from $\Wfin{S}$ by adding a~formal neutral element $\epsilon$. The initial state is $\epsilon$.
	For $v\in\Peps{\Wfin{S}}$ and $w\in\albet_W$ let
	\[\delta\big(v, w\big)\eqdef \big(f, v\cdot \alpha(w)\big),\]
	where $f\from \Winf{S}\to \albet_Y$ is defined for every $h\in \Winf{S}$ as follows: $f(h)$ is any fixed letter $y\in\albet_Y$ such that $v\cdot\beta(w,1)\cdot h\in F_y$,
	or just any element of $\albet_Y$ if $v\cdot\beta(w,1)\cdot h\not\in F_y$ for all $y\in\albet_Y$
	(morally, one should think that there is a unique such $y$;
	however strictly speaking this needs not to be true, which is caused by words $\word{w}$ for which $\exists Y.\,\psi(\word{w},Y)$ does not hold).

	Fix now an~input $\omega$\=/word $\word{w}=w_0w_1w_2\cdots\in(\albet_W)^\omega$ such that $\exists Y.\,\psi(\word{w},Y)$ holds.
	After reading a prefix $w_0w_1\cdots w_{n-1}$, the state of $\tau$ is $\alpha(w_0w_1\cdots w_{n-1})$ (or just $\epsilon$ if $n=0$).
	It follows that the $n$\=/th letter of $\funapp{\tau(\word{w})}{\lk_\alpha(\word{w})}$ is a~letter $y$ that satisfies
	\[\alpha(w_0w_1\cdots w_{n-1})\cdot\beta(w_n,1)\cdot\alpha(w_{n+1}w_{n+2}w_{n+3}\cdots)\in F_y.\]
	This is the case precisely when $\psi_y(\word{w},\chi_n)$ holds, and because $\psi$ is uniformised, this holds for precisely one $y$,
	which is the letter at position $n$ in the unique $\word{y}$ such that $\psi(\word{w},\word{y})$ holds.
	We thus obtain that $\psi\big(\word{w}, \funapp{\tau(\word{w})}{\lk_\alpha(\word{w})\big)}$ holds, as required.
\end{proof}

\subsection{Allow lookahead}

Using \cref{ft:uniformised-lookahead} we now show that a~winning strategy of Player~$\rII$ for a~game quantifier can be realised by a~transducer composed with a lookahead.

\begin{restatable}{lemma}{lemAllowLookahead}
\label{lem:allow-lookahead}
	Given a formula $\varphi(W,X,Y)$, one can effectively construct a~homomorphism $\alpha\from \Wilk{(\albet_W)}\to S$ onto a~finite Wilke algebra $S$
	together with a~transducer $\tau\from (\albet_W\times\albet_X)\tto (\Winf{S}\to \albet_Y)$ such that
	for every $\omega$\=/word $\word{w}\in(\albet_W)^\omega$ satisfying $\qG X\sto Y.\,\varphi(\word{w},X,Y)$, and for every $\word{x}\in(\albet_X)^\omega$ we have
	\[\varphi\big(\word{w},\word{x},\funapp{\tau(\word{w},\word{x})}{\lk_\alpha(\word{w})}\big).\]
\end{restatable}

In other words, if for every $n\in\N$ as $f_n\in(\Winf{S}\to\albet_Y)$ we take the output letter produced by $\tau$ after reading the prefixes $w_0w_1\dots w_n$ of $\word{w}$ and $x_0x_1\dots x_n$ of $\word{x}$
(so that $f_0f_1f_2\dots=\tau(\word{w},\word{x})$), and we consider $h_n\eqdef \alpha(w_{n+1}w_{n+2}w_{n+3}\ldots)$ and $y_n\eqdef f_n(h_n)$, then $\varphi(\word{w},\word{x},y_0y_1y_2\ldots)$ holds.
Intuitively, the above lemma says that one can construct the resulting $\omega$\=/word $\word{y}$ by a~transducer, assuming that we allow a~lookahead over the whole $\omega$\=/word $\word{w}$
(note that there is no lookahead over $\word{x}$: moves of Player~$\rII$ cannot be allowed to depend on future moves of Player~$\rI$).

This lemma is essentially a~composition of \cref{ft:uniformised-lookahead} with the following lemma. The only technical difficulty lies in the fact that the lookahead is given after the transducer has read the whole input $\omega$\=/word. A~complete proof of \cref{lem:allow-lookahead} is given in \cref{app:allow-lookahead}.

\begin{lemma}[\cite{lifsches_skolem,rabinovich_decidable,siefkes_monadic}]\label{lem:uniformisation}
	For every MSO formula $\psi(W,Y)$ one can construct a~uniformised formula $\psi_\mathsf{u}(W,Y)$ such that
	\begin{itemize}
	\item for all $\omega$\=/words $\word{w}$ we have $\big(\exists Y.\,\psi(\word{w},Y)\big)\Rightarrow(\exists Y.\,\psi_\mathsf{u}(\word{w},Y)\big)$, and
	\item for all $\omega$\=/words $\word{w}$, $\word{y}$ we have $\psi(\word{w},\word{y})\Leftarrow\psi_\mathsf{u}(\word{w},\word{y})$.
	\end{itemize}
\end{lemma}

It may be worth mentioning that the original proof of the above lemma as~given by Lifsches and Shelah~\cite[Theorem~6.3]{lifsches_skolem} says ``By [1].'',
where ``[1]'' is the work of B\"uchi and Landweber on synthesis~\cite{buchi_synthesis}.
This is incorrect, because B\"uchi and Landweber show how to win games using finite\=/state strategies, while a~uniformisation may in general depend on the future.
More precisely, it is always possible to uniformise a~formula $\psi(X,Y)$ even if $\qG X\sto Y.\,\psi(X,Y)$ does not hold
(imagine $\psi(X,Y)$ saying that the first letter of $Y$ is $1$ if and only if infinitely many letters of $X$ are $1$).
However, the mistake made by Lifsches and Shelah may not be a~coincidence: \cref{lem:allow-lookahead} shows that uniformisation is indeed possible using transducers with lookahead, while \cref{thm:transducer} proved later shows how to eliminate the lookahead when $\psi(W,X,Y)$ depends separately on the involved variable $Y$ (see \cref{def:separate}).

\subsection{Separate coordinates}
\label{ssec:separate}

\cref{ft:no-transducer} tells us that in general the lookahead $\lk_\alpha$ in \cref{lem:allow-lookahead} is necessary when we want to realise a~winning strategy by a transducer.
The example from \cref{ft:no-transducer} needed a~lookahead only to check the letter on the next position of the parameter $\word{w}$.
One can imagine another example: in order to win, Player~$\rII$ should output $y_n$ that equals the first letter in $\set{a,b}$ among $w_{n+1},w_{n+2},w_{n+3},\dots$
(skipping all letters $c$ before it).
Here the future interval checked by the lookahead needs to be unbounded, but still finite.
On the other hand, the lookahead never needs to check ``the whole infinite future'' of $\word{w}$, since $\word{w}$ is known in advance.
To see the intuitions for this, assume that the value of $y_n$ needed to win depends on whether letter $a$ belongs to the set $\{w_{n+1},w_{n+2},w_{n+3},\ldots\}$.
Here a~winning strategy seems to depend on the whole future of $\word{w}$, but since we know $\word{w}$ in advance, we may avoid this:
either $\word{w}$ contains infinitely many $a$ (and then there is always some $a$ in the future),
or the last $a$ occurs on some position $n$ (and then the transducer may count to the fixed number~$n$).
This suggests that it should be possible to create a~transducer which does not use a lookahead, but produces letters of $\word{y}$ with some delay, needed to check future properties of~$\word{w}$
(formally, this is impossible to realise, because there is also an issue of synchronisation between $\word{y}$ and $\word{x}$).

The main result of this section states that
we may avoid the aforementioned need for a~lookahead (or a~delay) once we disallow $\varphi$ to enforce any position\=/to\=/position correspondence between $\word{w}$ and $\word{y}$.
More precisely, we introduce the following definition.

\begin{definition}
\label{def:separate}
We say that a~formula $\varphi(W,X,Y)$ depends \emph{separately on~$Y$} if it is a~finite Boolean combination of formulae $\psi_i(W,X)$ and formulae $\gamma_i(Y)$.
\end{definition}

\begin{restatable}{theorem}{thmTransducer}
\label{thm:transducer}
	Assume that $\varphi(W,X,Y)$ depends separately on $Y$ and that $\word{w}\in(\albet_W)^\omega$ is such that $\qG X\sto Y.\,\varphi(\word{w},X,Y)$ holds.
	Then, there exists a~transducer $\tau_{\word{w}}\from\albet_W\times\albet_X\tto\albet_Y$ such that \cref{eq:transducer} holds, that is,
	for every $\word{x}\in(\albet_X)^\omega$ we have $\varphi\big(\word{w},\word{x},\tau_{\word{w}}(\word{w},\word{x})\big)$.
\end{restatable}

It is worth mentioning that if $\varphi(W,X,Y)$ is a~Boolean combination of formulae $\psi_i(W,X)$ and formulae $\gamma_i(X,Y)$
then one can still recover the example from \cref{ft:no-transducer} by writing that either $(\word{x}\neq\word{w})$ or $y_n=x_{n+1}$ for all $n\in\N$,
which is of the required shape and still no transducer can realise the strategy.

Since each transducer induces a~strategy, \cref{thm:transducer} in fact provides an~equivalence.

\begin{corollary}
\label{cor:strat-vs-trans}
Take $\varphi(W,X,Y)$ that depends separately on $Y$ and any $\word{w}\in(\albet_W)^\omega$. Then $\qG X\sto Y.\,\varphi(\word{w},X,Y)$ holds if and only if there exists a~transducer $\tau_{\word{w}}\from\albet_W\times\albet_X\tto\albet_Y$ such that for every $\word{x}\in(\albet_X)^\omega$ we have $\varphi\big(\word{w},\word{x},\tau_{\word{w}}(\word{w},\word{x})\big)$.
\end{corollary}

A~complete proof of \cref{thm:transducer} is given in \cref{app:transducer}, here we only provide some overview of the construction.
We begin by applying \cref{lem:allow-lookahead} to construct a~transducer which constructs the desired winning strategy.
The transducer uses lookahead given by a~homomorphism $\alpha\from\Wilk{(\albet_W)}\to S$ into some finite Wilke algebra $S$.
We fix the parameter~$\word{w}$ and apply Ramsey's theorem (cf.~\cref{thm:ramsey}) to split the word $\word{w}$ into a~finite prefix
and then infinitely many subwords whose image under $\alpha$ is the same idempotent $e\in\Wfin{S}$.
In all the split points the transducer can be sure that the lookahead (i.e., the value of the suffix under~$\alpha$) is $\alpha(e^\omega)$.
But how can the transducer detect the split points?
The first one can be hardcoded in the transducer.
Then, knowing some split point, the transducer has to find a~next one.
A brave conjecture would be that every subword evaluating to $e$ moves us from one splitting point to a next one;
but this is false (maybe we should split after a~different subword evaluating to $e$).
However, a~slightly stronger condition is sufficient: if~the transducer encounters two consecutive subwords evaluating to $e$,
then it can be sure that the position after the first of them can be chosen (that is, the suffix after this position can be split into infinitely many subwords evaluating to $e$).
We remark that a~similar technical trick occurs in a~work of Thomas~\cite{thomas_first_order}.

These infinitely many splitting points detected by the transducer (with some delay) are positions where we know the value of the lookahead,
and thus we can produce at those points the fragments of an~actual output $\omega$\=/word $\word{y}$.
Since these positions are scattered in an~arbitrary way, we need to be able to \emph{pad} the output in\=/between these positions.
This is where the assumption of separate dependency on $Y$ comes into play: the satisfaction of $\varphi(\word{w},\word{x},\word{y})$
depends separately on $\langle\word{w},\word{x}\rangle$ and on $\beta(\word{y})$ for an~appropriately chosen homomorphism $\beta\from\Wilk{(\albet_Y)}\to T$ onto another finite Wilke algebra $T$.
Now, using some standard techniques involving idempotents, we can ensure that we pad the output $\omega$\=/word in such a~way that there are no delays in its generation
(i.e., we really produce some letter from $\albet_Y$ in each step),
while still we control the final value of $\beta(\word{y})$, making sure that $\varphi(\word{w},\word{x},\word{y})$ holds.

\section{Index quantifiers over \texorpdfstring{$\omega$\=/}{w-}words}

We now use the previous results to show quantifier elimination procedure for index quantifiers.

\begin{theorem}
\label{thm:index-omega}
	The logic MSO+$\qI$ effectively reduces to the pure MSO over $\omega$\=/words.
\end{theorem}

First, one can observe that the non\=/deterministic index quantifiers $\qI^\Indet_\anIndex$ are either trivial or equivalent to the deterministic ones $\qI^\Idet_{\anIndex'}$, with an~appropriate change of indices. More precisely, the following equivalences hold:
\begin{itemize}
\item $\qI^\Indet_{\anIndex} X.\,\varphi(W,X)$ is equivalent to $\forall X.\, \varphi(W,X)$ whenever $\anIndex$ is either $\CP_{0,0}$ or $\CW_{0,0}$, because automata of these indices accept all $\omega$\-/words.
\item $\qI^\Indet_{\anIndex} X.\,\varphi(W,X)$ is equivalent to $\forall X.\, \lnot\varphi(W,X)$ whenever $\anIndex$ is either $\CP_{1,1}$ or $\CW_{1,1}$, because automata of these indices recognise empty languages.
\item $\qI^\Indet_{\CP_{1,2}} X.\, \varphi(W,X)$ is always true, because non\=/deterministic B\"uchi automata recognise all regular languages of $\omega$\=/words~\cite{buchi_decision}.
	The same holds for $\CP_{i,j}$ with $i\in\set{0,1}$ and $j\geq 2$.
\item $\qI^\Indet_{\CP_{0,1}} X.\, \varphi(W,X)$ is equivalent to $\qI^\Idet_{\CP_{0,1}} X.\, \varphi(W,X)$, because non\=/deterministic co\=/B\"uchi automata have the same expressive power as deterministic co\=/B\"uchi automata~\cite{miyano_threshold}.
\item $\qI^\Indet_{\anIndex} X.\, \varphi(W,X)$ is equivalent to $\qI^\Idet_{\anIndex} X.\, \varphi(W,X)$ whenever $\anIndex$ is either $\CW_{0,1}$ (safety), $\CW_{1,2}$ (reachability), or $\CW_{0,2}$, again because of the ability to determinise these automata without change of index (simple powerset\=/like constructions suffice).
\item $\qI^\Indet_{\CW_{1,3}} X.\, \varphi(W,X)$ is equivalent to $\qI^\Idet_{\CW_{0,1}} X.\, \varphi(W,X)$, because non\=/deterministic $\CW_{1,3}$ automata have the same expressive power as deterministic co\=/B\"uchi automata.
	The same holds for all $\CW_{i,j}$ with $i\in\set{0,1}$ and $j\geq 3$.
\end{itemize}
This covers all cases, because we can always shift the indices so that $i\in\set{0,1}$.

Thus, for the rest of this section we focus on the deterministic index quantifiers~$\qI_\anIndex^{\Idet}$.
Similarly as in \cref{cor:elim-game-quant} we proceed inductively, that is, we eliminate index quantifiers starting from inside.
Let $\anIndex=\langle\albet_\anIndex,L_\anIndex\rangle$ be an~index (either $\anIndex=\CP_{i,j}$, or $\anIndex=\CW_{i,j}$ for $i\leq j$).
Consider a~formula $\qI^{\Idet}_{\anIndex} X.\,\varphi(W,X)$ with $\varphi(W,X)$ in MSO.
Our goal is to construct a formula of pure MSO that is equivalent to $\qI^{\Idet}_{\anIndex} X.\,\varphi(W,X)$.
As in the previous section, we consider here only the case of a single parameter $W$
(which can encode multiple parameters using a~product alphabet).

Consider $\varphi_\anIndex(W,X,K)$ which, for $\word{w}\in(\albet_W)^\omega$, $\word{x}\in (\albet_X)^\omega$, and $\word{k}\in(\albet_\anIndex)^\omega$,
says that $\word{k}\in L_\anIndex$ if and only if $\varphi(\word{w},\word{x})$ holds.
Note that $\varphi_\anIndex(W,X,K)$ depends separately on the variable $K$.

\begin{lemma}\label{lem:index2game}
	The formulae $\qI^{\Idet}_{\anIndex} X.\,\varphi(W,X)$ and $\qG X \sto K.\,\varphi_\anIndex(W,X,K)$ are equivalent.
\end{lemma}

Note that this lemma concludes the proof of \cref{thm:index-omega} because the game quantifier~$\qG$ involved in the latter formula can be effectively eliminated due to \cref{lem:elim-1-game-quant}.

\begin{proof}
Take any parameter $\word{w}\in (\albet_W)^\omega$. The following conditions are equivalent:
\begin{itemize}
\item $\qI^{\Idet}_{\anIndex} X.\,\varphi(\word{w},X)$ holds;
\item there exists a~deterministic automaton $\CD_{\word{w}}$ of index $\anIndex$ such that for every $\omega$\=/word $\word{x}\in(\albet_X)^\omega$ we have $\varphi(\word{w},\word{x})$ if and only if $(\word{w},\word{x})\in\lang(\CD_{\word{w}})$;
\item there exists a~transducer $\tau_{\word{w}}\from (\albet_W\times\albet_X)\tto \albet_\anIndex$ such that for every $\omega$\=/word $\word{x}\in(\albet_X)^\omega$ we have $\varphi(\word{w},\word{x})$ if and only if $\tau_{\word{w}}(\word{w},\word{x})\in L_\anIndex$;
\item $\qG X\sto K.\,\varphi_\anIndex(\word{w},X,K)$ holds.
\end{itemize}

The first two items are equivalent from the definition of the index quantifier. The second two items are equivalent due to \cref{rem:aut-vs-trans}. The last two items are equivalent by the choice of $\varphi_\anIndex$ and \cref{cor:strat-vs-trans}.
\end{proof}

\section{Index quantifiers over trees}

In this section we prove the following theorem.

\begin{restatable}{theorem}{thmTressIndexUndec}
\label{thm:trees-index-undec}
The theory of MSO+$\qI$ over trees is undecidable. This holds even if we allow only the simplest possible index, namely the weak parity index $\CW_{0,1}$ (i.e.,~safety), and any type of determinism $D\in\{\Idet, \Indet\}$, thus we use only the index quantifier $\qI^D_{\CW_{0,1}}$.
\end{restatable}

To simplify the notations, in this section we consider only variables over the alphabet $\set{0,1}$; thus their valuations can be seen as sets, rather than functions from tree nodes to $\set{0,1}$. Thus, we can write $\tree{x}\cup \tree{y}$, $\tree{x}\subseteq\tree{y}$, etc.

A~set $I\subseteq \{\dL,\dR\}^\ast$ of tree nodes is called an~\emph{interval} if it is of the form $\{u\dR^i\mid 0\leq i\leq n\}$ for some \emph{topmost node} $u$ (denoted~$\topn(I)$), \emph{bottommost node} $u\dR^n$, and some \emph{length} $n\in\N$ (denoted~$\len(I)$).
Two intervals $I_1$, $I_2$ are \emph{independent} if their topmost nodes $u_1$, $u_2$ are such that $u_1\not\preceq u_2$ and $u_2\not\preceq u_1$ (i.e.,~none of them is a~descendant of the other).
A~\emph{union of independent intervals} is a~set~$\tree{z}$ that can be written as $\bigcup_{i\in J}I_i$, where $(I_i)_{i\in J}$ are pairwise independent intervals.
Note that the decomposition of such $\tree{z}$ into intervals is unique.
Moreover, it can be accessed in MSO (we can write in MSO things like ``$I$ is one of the intervals in $Z$'', etc.).

The crucial technical contribution is the following lemma.

\begin{restatable}{lemma}{lemBoundToSafety}
\label{lem:bound2safety}
	In MSO+$\qI^\Indet_{\CW_{0,1}}$ and in MSO+$\qI^\Idet_{\CW_{0,1}}$ over trees
	one can write a~formula $\varphi_\mathsf{b}(Z)$ such that for each union of independent intervals $\tree{z}$,
	the intervals in $\tree{z}$ have lengths bounded by some $n\in\N$ if and only if $\varphi_\mathsf{b}(\tree{z})$ holds.
\end{restatable}

\begin{proof}[Proof sketch]
As $\varphi_\mathsf{b}(Z)$ we take
\[\forall W\subseteq Z.\ \forall Y\subseteq Z.\ \qI^D_{\CW_{0,1}}X.\ \underbrace{\big(X\subseteq W\land\forall x\in X.\,\exists y\in Y.\,x\preceq y\big)}_{\psi(W,Y,X)},\]
where $D$ is either $\Idet$ or $\Indet$ (i.e.,~the index quantifier is either for deterministic or non\=/deterministic automata---both will work). The formula $\varphi_\mathsf{b}(Z)$, given a~set $\tree{z}$, expresses that for every choice of subsets $\tree{w},\tree{y}\subseteq \tree{z}$ there is a~safety automaton $\CA$ that checks whether a~given subset~$\tree{x}$ of $\tree{w}$ contains only points with a~descendant in~$\tree{y}$.

Note that the subformula $\psi(W,Y,X)$ does not depend on $\tree{z}$ so a~hypothetical safety automaton~$\CA$ does not have access to $\tree{z}$, even though the variable $Z$ is formally available in the scope of the subformula $\psi(W,Y,X)$ (one may ensure that $\tree{z}$ is not visible inside $\psi(W,Y,X)$ by artificially overshadowing the variable $Z$ by another quantifier $\exists Z.$ in front of $\qI^D_{\CW_{0,1}}$).

Suppose first that the intervals in some set $\tree{z}$ have lengths bounded by some $n\in\N$, and take any $\tree{w},\tree{y}\subseteq \tree{z}$.
The only elements of $\tree{y}$ that can be descendants of elements of $\tree{w}$ are elements of the same interval, hence they are located at most $n$ levels below.
Thus the property in question is recognised by the following deterministic safety automaton:
after every element of $\tree{x}$ (check if it belongs to $\tree{w}$ and) wait for $n$ levels on the branch going only right;
if no element of $\tree{y}$ was found, reject.
This shows that $\varphi_\mathsf{b}(\tree{z})$ holds (no matter whether we used $\qI^\Indet_{\CW_{0,1}}$ or $\qI^\Idet_{\CW_{0,1}}$ in its definition).
Note that the size of the constructed automaton depends on $n$.

	Suppose now that intervals in $\tree{z}$ have unbounded lengths.
	We want to prove that $\varphi_\mathsf{b}(\tree{z})$ does not hold.
As $\tree{w}$ we choose the topmost points of all intervals in $\tree{z}$.
The set $\tree{y}$ contains the bottommost points of carefully chosen intervals from $\tree{z}$. Among other properties, the chosen intervals need to have growing lengths.
First of all, the whole construction is done in a~diagonal way, because the hypothetical safety automaton $\CA$ is not known in advance, so we iterate over all such automata, ensuring that each of them is not a~good witness for $\qI^D_{\CW_{0,1}}X.\ \psi(\tree{w},\tree{y},X)$ to hold.

For a~fixed automaton $\CA$, if the lengths of the intervals are unbounded, then some interval exceeds the counting capacity of the automaton~$\CA$. Thus, some interval whose bottommost point $y$ belongs to $\tree{y}$ needs to be sufficiently long, so that the automaton $\CA$ admits a~pumping pattern with two repeating states along the interval. This allows to repeat the pattern indefinitely (i.e.,~\emph{pump}), which effectively removes the node $y$ from $\tree{y}$, making the formula $\psi(\tree{w},\tree{y},\tree{x})$ false, while the hypothetical automaton still accepts $\langle \tree{w},\tree{y},\tree{x}\rangle$. However, we are not allowed to change the parameter set $\tree{y}$, so the actual pumping needs to be performed on another part of $\tree{y}$, which requires to apply a~Ramsey\=/like argument allowing us to shift the pumping place outside the considered set $\tree{y}$ (this vaguely resembles the pumping scheme from Carayol and L\"oding~\cite{loding_choice}).

A~complete proof of this implication is given in \cref{appp:bound2safety}.
\end{proof}

Once we know that boundedness of sets of independent intervals is expressible in MSO+$\qI^D_{\CW_{0,1}}$, it remains to adjust the technical construction from Bojańczyk et al.~\cite{bojanczyk_msou_final} to express runs of Minsky machines in terms of boundedness of sets of intervals. This is a~rather standard adjustment; see \cref{app:trees-index-undex} for more details.

\section{Conclusions}

Our motivation in this work was to introduce and study index quantifiers $\qI$, which try to incorporate the index problem into the syntax of the logic. From that perspective,
% the most important result is
an~important message stems from
\cref{thm:trees-index-undec} stating that the logic MSO+$\qI$ is undecidable over trees, even in the simplest form of the quantifier, namely deterministic safety index quantifier.
The problem whether a~regular tree language can be recognised by an~automaton of index $\anIndex$ amounts to the satisfaction of a formula $\qI^D_\anIndex X.\,\varphi(X)$, in which the index quantifier is used only once and no parameters are allowed.
While this fragment can still be decidable, \cref{thm:trees-index-undec} shows that the frontier is close, and this is (to our knowledge) the first
undecidability result related to the index problem.

On the other hand, our study of index quantifiers~$\qI$ over $\omega$\=/words provides a~more optimistic view.
We have introduced an~effective index\=/quantifier elimination procedure, thus reducing MSO+$\qI$ into pure MSO.
Although the usefulness of index quantifiers in the realm of verification and model\=/checking is questionable,
our approach provides a~new and generic game\=/based way of proving that the index problem over $\omega$\=/words is decidable~\cite{wagner_hierarchy}.
This follows similar lines to Löding~\cite{loding_wadge_dpda}.
We believe that the most important consequence of this aspect of our study was a~thorough analysis of \emph{game quantifiers}~$\qG$---%
our proof of quantifier elimination for $\qI$ goes through a~translation to $\qG$ and then elimination of those.

Game quantifiers have been known and studied since 70s, however their applications were mostly limited to descriptive set theory.
In a~way similar to Kaiser~\cite{kaiser_game_quant}, we explicitly introduce the notion of game quantifiers~$\qG$ and study the expressive power of MSO+$\qG$.
The quantifier\=/elimination procedure for $\qG$ is quite direct and resembles a~parametrised version of the construction of B\"uchi and Landweber (see \cref{rem:effective-strategy}).
In contrast to index quantifiers, we believe that a~direct use of game quantifiers in MSO may be quite useful when expressing certain game\=/related properties.

For the sake of applying our quantifier\=/elimination procedure to index quantifiers, we needed to develop a~theory of transducers realising winning strategies in parametrised $\omega$\=/regular games.
To achieve it, we show a~novel and apparently quite strong result (see \cref{thm:transducer}) stating that if the involved variables are in some sense separate,
then whenever a~game quantifier holds, the existence of a~respective strategy can be witnessed by a~finite\=/memory transducer.

\bibliography{mskrzypczak}

\newpage\appendix

\section{Proof of \texorpdfstring{\cref{lem:allow-lookahead}}{Lemma \ref{lem:allow-lookahead}}}
\label{app:allow-lookahead}

\lemAllowLookahead*

\begin{proof}
	Recall the game $\theGame(\word{w},\CD)$ from \cref{def:product-game}, where $\CD$ is a~deterministic parity automaton equivalent to $\varphi$.
	As previously, let $Q$ be the set of states of~$\CD$.
	Take the formula $\psi^{(\rII)}(W,\Sigma)$ from \cref{claim:formula-for-strategy},
	which is such that $\psi^{(\rII)}(\word{w},\word{\sigma})$ holds if $\word{\sigma}\in(Q\times\albet_X\to \albet_Y)^\omega$ encodes a~positional winning strategy of Player~$\rII$ in $\theGame(\word{w},\CD)$.
	Using \cref{lem:uniformisation}, change $\psi^{(\rII)}(W,\Sigma)$ into a~uniformised formula $\psi^{(\rII)}_\mathsf{u}(W,\Sigma)$.
	For every $\omega$\=/word $\word{w}\in(\albet_W)^\omega$ satisfying $\qG X\sto Y.\,\varphi(\word{w},X,Y)$, such a positional winning strategy exists,
	hence $\psi^{(\rII)}_{\mathsf{u}}(\word{w},\word{\sigma})$ holds for precisely one $\omega$\=/word $\word{\sigma}\in(Q\times\albet_X\to \albet_Y)^\omega$,
	and this $\omega$\=/word encodes a~positional winning strategy of Player~$\rII$ in $\theGame(\word{w},\CD)$.
	Applying \cref{ft:uniformised-lookahead} to $\psi^{(\rII)}_\mathsf{u}$
	we obtain a~homomorphism $\alpha\from \Wilk{(\albet_W)}\to S$ onto a~finite Wilke algebra $S$ together with a~transducer
	\[\theta\from \albet_W\tto \big(\Winf{S}\to (Q\times\albet_X\to \albet_Y)\big)\]
	such that, for every $\omega$\=/word $\word{w}\in(\albet_W)^\omega$ satisfying $\qG X\sto Y.\,\varphi(\word{w},X,Y)$,
	the result $\funapp{\theta(\word{w})}{\lk_\alpha(\word{w})}$ encodes the unique winning strategy in $\theGame(\word{w},\CD)$.

	The homomorphism $\alpha$ is already satisfactory; it can be taken as the homomorphism $\alpha$ in the statement of \lcnamecref{lem:allow-lookahead}.
	On the other hand, the transducer has to be improved.
	The problem is that $\theta$ reads only $\word{w}$, and outputs a~strategy that depends on the state of $\CD$ and on the current letter of $\word{x}$;
	however $\tau$ should read both $\word{w}$ and $\word{x}$ and produce just the played letter of $\word{y}$.
	Thus, we enhance $\theta$ with a~part keeping track of the current state of $\CD$ based on the parts of $\word{w}$ and $\word{x}$ read so far,
	and on the previous letters of $\word{y}$ produced by $\theta$ itself.
	Note that the produced letters of $\word{y}$, hence also the current state of $\CD$, depend on results of the lookahead of the homomorphism $\alpha$,
	examining the future part of $\word{w}$; however values of the homomorphism in the past can be recovered from the current value.
	Formally, we construct the~new transducer $\tau\from (\albet_W\times\albet_X)\tto (\Winf{S}\to \albet_Y)$ as follows.

	Suppose that
	\begin{align*}
	\theta &=\big\langle \albet_W, \big(\Winf{S}\to (Q\times\albet_X\to \albet_Y)\big),P,\iota_\theta, \delta_\theta\big\rangle,\\
	\CD &=\big\langle(\albet_W\times\albet_X\times \albet_Y), \{i,i{+}1,\ldots,j\}, Q,\iota_\CD,\delta_\CD\big\rangle.
	\end{align*}
	As the set of states of $\tau$ we take the product $P\times (\Winf{S}\to Q)$.
	The initial state of $\tau$ is $(\iota_\theta,r_\init)$ where $r_\init(h)=\iota_\CD$ for all $h\in \Winf{S}$.
	In order to define the transition function $\delta_\tau$ of $\tau$, consider a state $(p,r)\in P\times(\Winf{S}\to Q)$ and an~input letter $(w,x)\in\albet_W\times\albet_X$.
	First, define $(\ell,p')\eqdef\delta_\theta(p,w)$;
	recall that $\ell\from \Winf{S}\to(Q\times\albet_X\to\albet_Y)$ is an~output letter produced by $\theta$.
	Next, for each $h'\in \Winf{S}$ define (as illustrated on \cref{fig:trans-flow})
	\begin{itemize}
	\item $h\eqdef \alpha(w)\cdot h'$,
	\item $q\eqdef r(h)$,
	\item $y\eqdef \ell(h')(q,x)$,
	\item $q'\eqdef \delta_\CD\big(q,(w,x,y)\big)$,
	\item $f(h')\eqdef y$,
	\item $r'(h')\eqdef q'$.
	\end{itemize}
	Having all this, we take
	\[\delta_\tau\big((p,r),(w,x)\big)=\big(f, (p', r')\big).\]

\tikzstyle{lookahead} = [line width=1pt, double distance=3pt, arrows = {-Latex[length=0pt 3 0]}]
\tikzstyle{transhead} = [bend left=15, -latex]
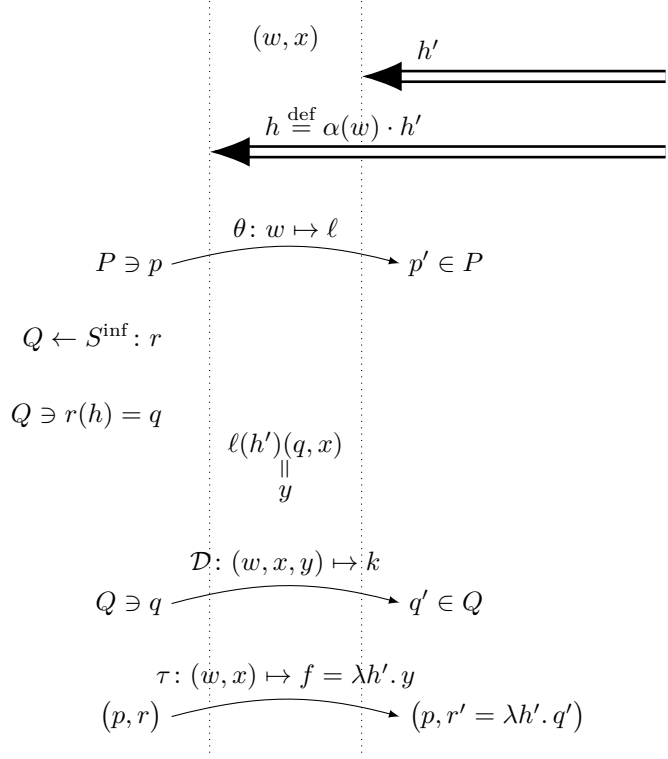
\begin{figure}
\centering
\begin{tikzpicture}
\draw (-1, 3) edge[dotted] (-1, -7);
\draw (+1, 3) edge[dotted] (+1, -7);

\node (w) at (0, 2.5) {$(w,x)$};

\draw (5, 2) edge[lookahead] (+1, 2);
\node[bzL] at (+1+0.6, 2.3) {$h'$};

\draw (5, 1) edge[lookahead] (-1, 1);
\node[bzL] at (-1+0.6, 1.3) {$h\eqdef \alpha(w)\cdot h'$};

\node[bzR] (p0) at (-1.5, -0.5) {$P \ni p$};
\node[bzL] (p1) at (+1.5, -0.5) {$p'\in P$};
\draw (p0.east) edge[transhead] node[above] {$\theta\from w\mapsto \ell$} (p1.west);

\node[bzR] (r0) at (-1.5, -1.5) {$Q\ot \Winf{S}\colon r$};
\node[bzR] (q0) at (-1.5, -2.5) {$Q \ni r(h)=q$};

\node[bzC] (y) at (0, -2.9) {$\ell(h')(q,x)$};
\node[bzC, rotate=90] (y) at (0, -3.2) {$=$};
\node[bzC] (y) at (0, -3.5) {$y$};

\node[bzR] (q1) at (-1.5, -5) {$Q \ni q$};
\node[bzL] (q2) at (+1.5, -5) {$q'\in Q$};
\draw (q1.east) edge[transhead] node[above] {$\CD\from (w,x,y)\mapsto k$} (q2.west);

\node[bzR] (a0) at (-1.5, -6.5) {$\big(p,r\big)$};
\node[bzL] (a1) at (+1.5, -6.5) {$\big(p,r'=\lambda h'.\,q'\big)$};
\draw (a0.east) edge[transhead] node[above] {$\tau\from (w,x)\mapsto f=\lambda h'.\,y$} (a1.west);
\end{tikzpicture}
\caption{An illustration of the flow of information in the proof of \cref{lem:allow-lookahead}. We use the notation $q\tran{\tau\from w\mapsto y} q'$ to indicate that a~transducer $\tau$ has a~transition $\delta(q,w)=(y,q')$. The types of suffixes of the $\omega$\=/word $\word{w}$ given by $\lk_\alpha(\word{w})$ are drawn using double arrows.}
\label{fig:trans-flow}
\end{figure}
	Consider now input words $\word{w}=w_0w_1w_2\cdots\in(\albet_W)^\omega$ and $\word{x}=x_0x_1x_2\cdots\in(\albet_X)^\omega$.
	Let $\word{\ell}=\ell_0\ell_1\ell_2\cdots\eqdef\theta(\word{w})$ be the output of $\theta$ and let $\word{f}=f_0f_1f_2\cdots\eqdef\tau(\word{w},\word{x})$ be the output of $\tau$.
	Moreover, for every $n\in\N$, let $h_n\eqdef\alpha(w_{n+1}w_{n+2}w_{n+3}\dots)$, let $y_n\eqdef f_n(h_n)$,
	and let $(p_n,r_n)$ be the state of $\tau$ after reading the word $(w_0,x_0)(w_1,x_1)\cdots(w_{n-1},x_{n-1})$.
	First, we observe by induction on $n\in\N$ that each $p_n$ is the state of $\theta$ after reading the word $w_0w_1\cdots w_{n-1}$;
	this is clear, because $p_{n+1}$ is defined as the second coordinate of $\delta_\theta(p_n,w_n)$.

	Next, we observe by induction on $n\in\N$ that $q_n=r_n(\alpha(w_n)\cdot h_n)$ and that $y_n=\ell_n(h_n)(q_n,x_n)$.
	Indeed, clearly $q_0=\iota_\CD=r_\init(\alpha(w_0)\cdot h_0)=r_0(\alpha(w_0)\cdot h_0)$.
	Then, suppose that $q_n=r_n(\alpha(w_n)\cdot h_n)$ for some $n\in\N$.
	We have
	\[y_n=f_n(h_n)=\ell_n(h_n)(r_n(\alpha(w_n)\cdot h_n),x_n)=\ell_n(h_n)(q_n,x_n),\]
	where the first equality holds by the definition of $y_n$,
	the second by the definition of $\delta_\tau$, because $f_n$ is the output of $\tau$ produced while reading $(w_n,x_n)$ from state $(p_n,r_n)$,
	and the third by the induction hypothesis.
	Having this, we obtain
	\begin{align*}
		q_{n+1}&=\delta_\CD(q_n,(w_n,x_n,y_n))=\delta_\CD(q_n,(w_n,x_n,\ell_n(h_n)(q_n,x_n)))\\
			&=r_{n+1}(h_n)=r_{n+1}(\alpha(w_{n+1})\cdot h_{n+1}),
	\end{align*}
	where the first equality holds by the definition of $q_{n+1}$, the second by the previously shown equality about $y_n$,
	the third by the definition of $\delta_\tau$, because $r_{n+1}$ is the second coordinate of the state of $\tau$ reached after reading $(w_n,x_n)$ from the state $(p_n,r_n)$,
	and the last by the definition of $h_n$ and by the fact that $\alpha$ is a homomorphism.

	Finally, the equality $y_n=\ell_n(h_n)(q_n,x_n)$, obtained for all $n\in\N$, implies that if in $\theGame(\word{w},\CD)$ Player~$\rI$ plays the $\omega$\=/word $\word{x}$,
	then according to the strategy encoded by $\word{\sigma}=\funapp{\theta(\word{w})}{\lk_\alpha(\word{w})}$ Player~$\rII$ answers with the $\omega$\=/word $\word{y}\eqdef y_0y_1y_2\cdots=\funapp{\tau(\word{w},\word{x})}{\lk_\alpha(\word{w})}$.
	The aforementioned strategy is winning, thus $\varphi(\word{w},\word{x},\word{y})$ holds.
\end{proof}

\section{Proof of \texorpdfstring{\cref{thm:transducer}}{Theorem \ref{thm:transducer}}}
\label{app:transducer}

\thmTransducer*

For the whole proof assume that $\word{w}\in(\albet_W)^\omega$ is fixed and such that $\qG X\sto Y.\,\varphi(\word{w},X,Y)$ holds.

Let $\beta\from\Wilk{(\albet_Y)}\to T$ be a homomorphism onto a~finite Wilke algebra $T$ recognising all the languages $\gamma_i(Y)$ in the formula $\varphi(W,X,Y)$, as given by \cref{thm:recognition}.
We then have $\varphi(\word{w},\word{x},\word{y})\Leftrightarrow\varphi(\word{w},\word{x},\word{y}')$ as long as $\beta(\word{y})=\beta(\word{y}')$.

Recall the monoid $\Peps{\Wfin{T}}$ obtained by adding a~formal neutral element $\epsilon$ to $\Wfin{T}$.
As a first step towards constructing $\tau_{\word{w}}$, in \cref{lem:safety-delayed} we construct a~transducer $\theta_{\word{w}}$ that outputs values from $\Peps{\Wfin{T}}$
(which aim to represent values of infixes of $\word{y}$ under the homomorphism~$\beta$) instead of concrete letters of $\word{y}$.
Thanks to such a change, it is easier to implement the aforementioned delay.
Later, in \cref{lem:elim-neutral,lem:monoid2letter} we show how to change $\theta_{\word{w}}$ into the desired transducer $\tau_{\word{w}}$.

For a saturated $\omega$\=/word $\word{t}\in(\Peps{\Wfin{T}})^\omega$
we write $\varphi(\word{w},\word{x},\word{t})$ to say that any $\omega$\=/word $\word{y}\in(\albet_Y)^\omega$ such that $\beta(\word{y})=\odot(\word{t})$ satisfies $\varphi(\word{w},\word{x},\word{y})$.
Moreover, for $r\in\N$ we say that such $\word{t}$ is \emph{$r$\=/full} if the first $r$ letters of $\word{t}$ are from $\Wfin{T}$ (i.e.,~they are not $\epsilon$).

\begin{lemma}\label{lem:safety-delayed}
	For every $r\in\N$ one can effectively construct a~transducer $\theta_{\word{w}}\from \albet_W\times\albet_X\tto \Peps{\Wfin{T}}$
	such that for every $\word{x}\in(\albet_X)^\omega$ the output $\theta_{\word{w}}(\word{w},\word{x})$ is saturated and $r$\=/full,
	and $\varphi(\word{w},\word{x},\theta_{\word{w}}(\word{w},\word{x}))$ holds.
\end{lemma}

\begin{proof}
	First, we use \cref{lem:allow-lookahead} to obtain a~homomorphism $\alpha\from\Wilk{(\albet_W)}\to S$ onto a~finite Wilke algebra $S$
	together with a~transducer $\tau_Y\from (\albet_W\times\albet_X)\tto (\Winf{S}\to \albet_Y)$ such that for every $\word{x}\in(\albet_X)^\omega$ we have $\varphi(\word{w},\word{x},\funapp{\tau_Y(\word{w},\word{x})}{\lk_\alpha(\word{w})})$.
	Note that we now have two Wilke algebras: $S$ for ``source'' letters and $T$ for ``target'' letters.

	We begin by composing the outputs of $\tau_Y$ (i.e.,~letters in $(\Winf{S}\to \albet_Y)$) with $\beta\from \albet_Y\to\Wfin{T}$. The resulting transducer is
	$\tau_T\from(\albet_W\times\albet_X)\tto(\Winf{S}\to \Peps{\Wfin{T}})$.
	Denoting $\word{y}=y_0y_1y_2\cdots\eqdef\funapp{\tau_Y(\word{w},\word{x})}{\lk_\alpha(\word{w})}$ and $\word{t}=t_0t_1t_2\cdots\eqdef\funapp{\tau_T(\word{w},\word{x})}{\lk_\alpha(\word{w})}$ we thus have $t_i=\beta(y_i)$ for all $i\in\N$.
	Because $\beta$ is a~homomorphism, \cref{eq:assoc} applies and therefore $\beta(\word{y})=\odot(\word{t})$, which means that $\varphi(\word{w},\word{x},\funapp{\tau_T(\word{w},\word{x})}{\lk_\alpha(\word{w})})$ holds.
	The transducer $\theta_{\word{w}}$, which we construct, will output an~$\omega$\=/word $\word{t}'$ such that $\odot(\word{t}')=\odot(\word{t})$;
	this will imply that $\varphi(\word{w},\word{x},\theta_{\word{w}}(\word{w},\word{x}))$ holds as well.

	It remains to construct $\theta_{\word{w}}$ out of $\tau_T$ and $\alpha$.
	Recall that the $\omega$\=/word $\word{w}=w_0w_1w_2\cdots$ is fixed.
	To shorten the notation, for $i,j\in\N$ with $i<j$ we define $\alpha_{\word{w}}(i,j)\eqdef\alpha(w_{i+1}w_{i+2}\cdots w_j)\in\Wfin{S}$,
	as well as $\alpha_{\word{w}}(i,i)\eqdef\epsilon\in\Peps{\Wfin{S}}$ and $\alpha_{\word{w}}(i,\infty)\eqdef\alpha(w_{i+1}w_{i+2}w_{i+3}\cdots)\in\Winf{S}$.
	Given $e\in\Wfin{S}$, we say that a~position $i$ is \emph{$e$\=/splittable} if there are positions $i=i_0<i_1<i_2<\dots$ such that $\alpha_{\word{w}}(i_k,i_\ell)=e$ for all $k,\ell\in\N$ with $k<\ell$
	(in other words, if the suffix of $\word{w}$ starting after position $i$ can be split into subwords such that the image under $\alpha$ of each of them is $e$).
	By the Ramsey theorem (\cref{thm:ramsey}; see also~\cite[Theorem~2.1, page~78]{perrin_pin_words}), there exists a~position $i_0$ and a~value $e\in\Wfin{S}$ such that $i_0$ is $e$\=/splittable and $i_0\geq r$,
	where $r$ is the number from the statement of the \lcnamecref{lem:safety-delayed}.
	We fix $i_0$ and $e$ for the rest of the proof.
	Note that $e$ is an idempotent, that is, $e\cdot e=e$.

	The overall idea of $\theta_{\word{w}}$ is that it cumulates (multiplies) the output letters of $\tau_T$ until reaching the next $e$\=/splittable position.
	This is needed because $\tau_T$ has access to a~lookahead on $\alpha$, while $\theta_{\word{w}}$ does not,
	but on each $e$\=/splittable position the transducer $\theta_{\word{w}}$ can be sure that the value of $\alpha$ on the suffix is $e^\omega$
	(and values of $\alpha$ on finite subwords ending in an $e$\=/splittable position can be easily computed by $\theta_{\word{w}}$).

	How $\theta_{\word{w}}$ can know whether its current position is $e$\=/splittable?
	As a~first $e$\=/splittable position it can remember $i_0$.
	Then, knowing some $e$\=/splittable position, the transducer has to find a~next one.
	Of course, $\theta_{\word{w}}$ can maintain the value of the subword starting at the previous $e$\=/splittable position, and check whether it equals $e$.
	A brave conjecture would be that if $i$ is $e$\=/splittable and $\alpha_{\word{w}}(i,j)=e$, then $j$ is $e$\=/splittable as well---but this is false.
	However, a slightly stronger condition is sufficient:
	if the transducer encounters two consecutive subwords evaluating to $e$,
	then it can be sure that the position after the first of them is $e$\=/splittable, as described by the following claim.

	\begin{claim}\label{claim:split}
		Let $i<j<k$.
		If a~position $i$ is $e$\=/splittable and $\alpha_{\word{w}}(i,j)=\alpha_{\word{w}}(j,k)=e$, then $j$ is $e$\=/splittable.
	\end{claim}

	\begin{claimproof}
		Because $i$ is $e$\=/splittable, there exists an $e$\=/splittable position $\ell>k$ such that $\alpha_{\word{w}}(i,\ell)=e$
		(in the split of the suffix starting at $i$ into subwords evaluating to $e$, as $\ell$ we take any position greater than $k$).
		We then have
		\begin{align*}
			\alpha_{\word{w}}(j,\ell)&=\alpha_{\word{w}}(j,k)\cdot \alpha_{\word{w}}(k,\ell)=e\cdot \alpha_{\word{w}}(k,\ell)=e\cdot e\cdot \alpha_{\word{w}}(k,\ell)\\
			&=\alpha_{\word{w}}(i,j)\cdot \alpha_{\word{w}}(j,k)\cdot \alpha_{\word{w}}(k,\ell)=\alpha_{\word{w}}(i,\ell)=e,
		\end{align*}
		which means that $j$ is $e$\=/splittable.
	\end{claimproof}

	This construction is vaguely based on the ideas from Thomas~\cite{thomas_first_order}, where a~similar trick is introduced.
	Translating it to our context, it allows one to express the value $\alpha(\word{w})\in\Winf{S}$ in first\=/order logic equipped with predicates $\psi_s(i,j)$ for $s\in\Wfin{S}$ which state that $\alpha_{\word{w}}(i,j)=s$.

	We now come into details of the construction of $\theta_{\word{w}}$.
	An~\emph{open sector} is a pair $(s,g)$, where $s\in\Peps{\Wfin{S}}$ and $g\from \Winf{S}\to \Peps{\Wfin{T}}$.
	Given also $\word{x}=x_0x_1x_2\cdots\in(\albet_X)^\omega$, and denoting $f_0f_1f_2\cdots\eqdef\tau_T(\word{w},\word{x})$ and $t_0t_1t_2\cdots\eqdef\funapp{\tau_T(\word{w},\word{x})}{\lk_\alpha(\word{w})}$,
	we say that $(s,g)$ \emph{describes} a~pair of positions $(i,k)$ (with $i\leq k$) if $s=\alpha_{\word{w}}(i,k)$ and $t_{i+1}\cdot t_{i+2}\cdot\ldots\cdot t_k=g(\alpha_{\word{w}}(k,\infty))$.
	An~\emph{update} of $(s,g)$ by a~pair $(w,f)\in\albet_W\times(\Winf{S}\to\Wfin{T})$ is the open sector $(s\cdot\alpha(w),g')$
	with $g'(h)\eqdef g(\alpha(w)\cdot h)\cdot f(h)$ for all $h\in\Winf{S}$.
	We see that if $(s,g)$ describes $(i,k)$, and we update it by $(w_{k+1},f_{k+1})$, then the updated open sector describes $(i,k{+}1)$.
	The \emph{empty open sector} is $(\epsilon,g_\epsilon)$, where $g_\epsilon(h)=\epsilon$ for all $h\in\Winf{S}$.
	It describes the pair $(i,i)$ for every $i\in\N$.

	A \emph{closed sector} is a triple $(t,s,g)$, where $t\in\Wfin{T}$ and $(s,g)$ is an~open sector.
	Given also $\word{x}=x_0x_1x_2\cdots\in(\albet_X)^\omega$, and denoting $f_0f_1f_2\cdots\eqdef\tau_T(\word{w},\word{x})$ and $t_0t_1t_2\cdots\eqdef\funapp{\tau_T(\word{w},\word{x})}{\lk_\alpha(\word{w})}$,
	we say that $(t,s,g)$ \emph{describes} a~triple of positions $(i,j,k)$ (with $i<j\leq k$) if $\alpha_W(i,j)=e$, and $t_{i+1}\cdot t_{i+2}\cdot\ldots\cdot t_j=t$, and
	the open sector $(s,g)$ describes the pair $(j,k)$.
	An~\emph{update} of $(t,s,g)$ by a~pair $(w,f)\in\albet_W\times(\Winf{S}\to\Wfin{T})$ is the closed sector $(t,s',g')$, where $(s',g')$ is the update of $(s,g)$ by $(w,f)$.
	Again, if $(t,s,g)$ describes $(i,j,k)$, and we update it by $(w_{k+1},f_{k+1})$, then the updated closed sector describes $(i,j,k+1)$.

	While reading letters at positions $0,1,\dots,i_0$, the new transducer $\theta_{\word{w}}$ directly simulates $\tau_T$, and counts up to $i_0$;
	when $\tau_T$ wants to output a~letter $f\from\Winf{S}\to\Wfin{T}$ while reading a letter at a~position $i\in\set{0,1,\dots,i_0}$, then
	$\theta_{\word{w}}$ outputs $f(\alpha(w_{i+1}w_{i+2}w_{i+3}\cdots))$ (where $\alpha(w_{i+1}w_{i+2}w_{i+3}\cdots)$ can be hardcoded into the transducer, as $\word{w}$ is known in advance).
	This way, we can be sure that the output $\omega$\=/word $\theta_{\word{w}}(\word{w},\word{x})$ is $r$\=/full (recall that $i_0\geq r$),
	and that the letters of $\theta_{\word{w}}(\word{w},\word{x})$ up to position $i_0$
	are the same as in $\funapp{\tau_T(\word{w},\word{x})}{\lk_\alpha(\word{w})}$, for every $\word{x}\in(\albet_X)^\omega$.

	After reaching the position $i_0$, the transducer $\theta_{\word{w}}$ maintains in its state the following information:
	\begin{itemize}
	\item	the current state $q$ of $\tau_T$,
	\item	one open sector $o$, and
	\item	a set $C$ of closed sectors.
	\end{itemize}
	At the position $i_0$ the state of $\tau_T$ is known from the previous phase;
	as the open sector we take the empty open sector, and as $C$ we take the empty set.

	When a new input letter $(w,x)\in\albet_W\times\albet_X$ comes, the state of $\theta_{\word{w}}$ is updated as follows:
	\begin{enumerate}
	\item	The state of $\tau_T$ is updated according to its transition function; let $f\in(\Winf{S}\to\Wfin{T})$ be the output letter produced by $\tau_T$.
	\item	The open sector $o$ and all closed sectors in $C$ are updated by $(w,f)$.
		It is possible that different closed sectors become identical after the update; we store the result only once, since $C$ is a set.
	\item	If the open sector $o$ is of the form $(e,g)$, with $e$ on the first coordinate, we add $(g(e^\omega),\epsilon,g_\epsilon)$ to $C$,
		where $(\epsilon,g_\epsilon)$ is the empty open sector.
	\item	We check if in $C$ there is a closed sector $(t,e,g)$, with $e$ on the second coordinate.
		If so, we output the letter $t$, we replace the open sector $o$ by $(e,g)$, and we replace $C$ by $\emptyset$.
		If there are multiple such closed sectors in $C$, we choose only one of them arbitrarily, and we do the above.
		If no such closed sector is found, we output the neutral element $\epsilon\in\Peps{\Wfin{T}}$.
	\end{enumerate}
	We claim that the state of $\CB$ satisfies the following invariant, after reading the input $\omega$\=/word up to position $k\geq i_0$:
	there exist indices $i$, $j_0$ with $i_0\leq i\leq j_0\leq k$ such that
	\begin{enumerate}[(a)]
	\item	the position $i$ is $e$\=/splittable,
	\item	the open sector $o$ describes $(i,k)$,
	\item	each of the closed sectors in $C$ describes $(i,j,k)$ for some $j$ with $i<j\leq k$, and
	\item	for every position $j$ such that $j_0<j\leq k$ and $\alpha_{\word{w}}(i,j)=e$, there is a closed sector in $C$ that describes $(i,j,k)$.
	\end{enumerate}
	It is clear that this invariant is satisfied at the beginning, for $k=i_0$, where as $i$ and $j_0$ we take $i_0$.
	Let us now see that it is preserved while reading the letter at some position $k+1$.
	Suppose first that the result of the check in Point~4 is negative.
	We then leave $i$ and $j_0$ unchanged.
	Item~(a) is trivially preserved.
	Items~(b) and~(c) are preserved due to the update applied in Point~2; likewise Item~(d) with respect to positions $j\leq k$.
	If $\alpha_{\word{w}}(i,k+1)=e$, Item~(d) requires now also a~closed sector that describes $(i,k{+}1,k{+}1)$; this is the sector added to $C$ in Point~3.

	Suppose now that the check from Point~4 encounters in $C$ a~closed sector $(t,e,g)$, which for some $j$ with $i<j\leq k$ describes $(i,j,k{+}1)$
	(with $k{+}1$ on the last coordinate, because we are already after the update from Point~2).
	We thus have $\alpha_{\word{w}}(i,j)=\alpha_{\word{w}}(j,k{+}1)=e$, so by \cref{claim:split} we obtain that $j$ is $e$\=/splittable.
	Then as the new $i$ we take $j$, and as the new $j_0$ we take $k{+}1$.
	The above shows Item~(a), while Items~(b)--(d) are trivially satisfied.

	Notice moreover that when $i$ is shifted to position $j$, then we output $t$,
	which is the product of letters of $\funapp{\tau_T(\word{w},\word{x})}{\lk_\alpha(\word{w})}$ at positions $i{+}1,i{+}2,\ldots,j$;
	and if $i$ is not shifted, then we output $\epsilon$.
	Thus, if $i$ is shifted infinitely often, then the output of $\theta_{\word{w}}$ is saturated,
	and its infinite product $\odot\big(\theta_{\word{w}}(\word{w},\word{x})\big)$ is indeed the same as the infinite product $\odot\big(\tau_T(\word{w},\word{x})\circ\alpha(\word{w})\big)$.

	It remains to see that indeed $i$ is shifted infinitely often.
	Suppose it is not the case, and consider the last values of $i$ and $j_0$ (note that $j_0$ changes only together with $i$, and that the position on which the last change happened is exactly $j_0$).
	Consider a split of the suffix starting after position $i$ into subwords evaluating to $e$;
	it exists by Item~(a).
	Take any two positions $j$, $k{+}1$ of this split, greater than $j_0$, so that $j_0<j<k{+}1$.
	Because $\alpha_{\word{w}}(i,j)=e$, Item~(d) of the invariant at the position $k$ says that in $C$ we have a~closed sector describing $(i,j,k)$.
	At the position $k{+}1$ it becomes a closed sector describing $(i,j,k{+}1)$.
	Because $\alpha_{\word{w}}(j,k{+}1)=e$, it then has $e$ on the second coordinate.
	It is thus found by Point~4, which shifts $i$ to a higher value, contrarily to our assumption.
\end{proof}

In order to obtain \cref{thm:transducer} from \cref{lem:safety-delayed}, we now only need to slightly massage the obtained transducer, using some algebraic properties of Wilke algebras.
This is provided by the next two lemmata.

\begin{lemma}\label{lem:elim-neutral}
	Let $T$ be finite Wilke algebra, and $\Peps{\Wfin{T}}$ the monoid obtained from $\Wfin{T}$ by adding a formal neutral element $\epsilon$.
	There exists a~constant $r\in\N$ and a~transducer $\theta_\mathsf{d}\from\Peps{\Wfin{T}}\tto\Wfin{T}$
	such that for every saturated $r$\=/full $\omega$\=/word $\word{t}\in (\Peps{\Wfin{T}})^\omega$ we have
	\[\odot (\word{t}) = \odot\big(\theta_\mathsf{d}(\word{t})\big).\]
\end{lemma}

\begin{proof}
	As $r$ we take the constant from \cref{thm:ramsey-semi} for the semigroup $\Wfin{T}$.

	The transducer $\theta_\mathsf{d}$ works in two phases. The states of the first phase are of the form $c\in\Peps{\Wfin{T}}$, while the states of the second phase
	are of the form $(e,c)$ where $e\in\Wfin{T}$ and $c\in\Peps{\Wfin{T}}$.
	In the first phase, used on the prefix the input $\omega$\=/word, $\theta_\mathsf{d}$ copies the input letters
	and stores in its state the product $c\in\Peps{\Wfin{T}}$ of the word read so far.
	Then it moves to the second phase whenever there exists $e\in\Wfin{T}$ such that the new value of $c$ satisfies $c\cdot e=c$
	(if there are multiple such elements $e$, we choose one of them arbitrarily).
	Thus the initial state is $\epsilon$, and the transition function during the first phase from a~state $c\in \Peps{\Wfin{T}}$ is defined for $t\in\Wfin{T}$ by
	\begin{align*}
		\delta(c,t) &\eqdef \left\{\begin{array}{ll}
		\!\!\!\big(t, (e,\epsilon)\big)&\text{if $\exists e\in \Wfin{T}.\ c\cdot t\cdot e = c\cdot t$,}\\
		\!\!\!\big(t,c\cdot t\big) &\text{otherwise,}
	\end{array}\right.
	\end{align*}
	and $\delta(c,\epsilon)$ leads to a special ``error'' state, from which $\theta_\mathsf{d}$ constantly produces some fixed output letter.
	Note that after reading a~prefix $t_0t_1\cdots t_{j-1}$ with all letters in $\Wfin{T}$ (i.e.,~other than $\epsilon$),
	if the transducer is still in the first phase then its state is $t_0\cdot t_1\cdot \ldots \cdot t_{j-1}$ and the produced output is just a copy of the input.

	We now use \cref{thm:ramsey-semi} to see that while reading any $r$\=/full $\omega$\=/word $\word{t}=t_0t_1t_2\cdots\in (\Peps{\Wfin{T}})^\omega$,
	the transducer indeed enters the second phase somewhere over the first $r$ positions, avoiding the error state.
	Indeed, \cref{thm:ramsey-semi} applied to the prefix $t_0t_1\cdots t_{r-1}$ gives us a position $j$ and an idempotent $e$ such that (at latest) after reading $t_0t_1\cdots t_j$
	the current product $c\cdot t_j=t_0\cdot t_1\cdot\ldots\cdot t_j$ satisfies $c\cdot t_j\cdot e = c\cdot t_j$,
	which finishes the first phase.

	The transition function for this second phase is defined by
	\[\delta\big((e,c),t\big)\eqdef \left\{\begin{array}{lll}
		\!\!\!\big(c\cdot t, (e', \epsilon)\big)&\text{if $\exists e'\in \Wfin{T}.\ c\cdot t\cdot e' = c\cdot t$,}&\mbox{(Case I)}\\
		\!\!\!\big(e, (e, c\cdot t)\big)&\text{otherwise,}&\mbox{(Case II)}
	\end{array}\right.\]
	where if there are multiple elements $e'\in \Wfin{T}$, then we arbitrarily choose any of them.
	Note that $\epsilon\cdot e'\neq\epsilon$ for $e'\in\Wfin{T}$, so the output letter $c\cdot t$ in Case I is indeed in $\Wfin{T}$.

	Consider now some actual saturated $r$\=/full input $\omega$\=/word $\word{t}=t_0t_1t_2\cdots$, and a~position~$j$ after which the transducer is already in the second phase, in a~state $(e,c)$. Let $s_0s_1\cdots s_j$ be the output word produced so far.
	We keep the following invariant: there exists $j_0\leq j$ such that
	\begin{itemize}
	\item	$t_0\cdot t_1\cdot\ldots\cdot t_{j_0} = s_0\cdot s_1\cdot \ldots \cdot s_j$,
	\item	$t_0\cdot t_1\cdot\ldots\cdot t_{j_0} = t_0\cdot t_1\cdot\ldots\cdot t_{j_0}\cdot e$,
	\item	$t_{j_0+1}\cdot t_{j_0+2}\cdot\ldots\cdot t_j=c$.
	\end{itemize}
	The invariant clearly holds when we enter phase two from phase one, with $j_0=j$.
	Suppose now that a letter $t_{j+1}$ is read.
	If the transition function uses Case~I, then the invariant holds with $j_0$ shifted to $j{+}1$;
	after outputting $c\cdot t_{j+1}$, the product of output letters becomes multiplied by $c\cdot t_{j+1}=t_{j_0+1}\cdot t_{j_0+2}\cdot\ldots\cdot t_{j+1}$,
	which becomes compensated by the shift of $j_0$.
	If Case~II is used, the invariant is preserved with the same $j_0$; note that outputting $e$ we do not change the product of the output, because of the second item of the invariant.

	It remains to see that Case~I is used infinitely often, so that $j_0$ tends to infinity.
	To the contrary, assume that from some moment on, the value of $j_0$ remains unchanged.
	Take the first $r$ positions $i_0<i_1<\ldots<i_{r-1}$ such that $j_0<i_0$ and $t_{i_0},t_{i_1},\ldots,t_{i_{r-1}}\neq\epsilon$ (there are infinitely many such positions, because $\word{t}$ is saturated),
	and consider the word $t_{i_0}t_{i_1}\cdots t_{i_{r-1}}$.
	By \cref{thm:ramsey-semi} there is an index $j$ and an idempotent $e'$ such that $c\cdot t_{i_j}\cdot e' = c\cdot t_{i_j}$ for $c=t_{i_0}\cdot t_{i_1}\cdot\ldots\cdot t_{i_{j-1}}=t_{j_0+1}\cdot t_{j_0+2}\cdot\ldots\cdot t_{i_j-1}$.
	This means that Case~I is used while reading position $i_j$, contrarily to our assumption.
\end{proof}

\begin{lemma}\label{lem:monoid2letter}
	Assume that $\beta\from \Wilk{\albet}\to T$ is a~homomorphism onto a~finite Wilke algebra $T$.
	Then there exists a~transducer $\theta_\beta\from \Wfin{T}\tto \albet$ such that for every $\omega$\=/word $\word{t}\in (\Wfin{T})^\omega$ we have
	\[\odot (\word{t}) = \beta\big(\theta_\beta(\word{t})\big).\]
\end{lemma}

In this \lcnamecref{lem:monoid2letter} it is important that $\beta$ is indeed \emph{onto} $T$, that is, all elements of $T$ are in the image of $\beta$.
\begin{proof}
	First, for every $t\in\Wfin{T}$ we fix any word $u_t\in \albet^{+}$ such that $\beta(u_t)=t$.
	Then, as the set of states of $\theta_\beta$ we take
	\[\{(c, w)\in \Peps{\Wfin{T}}\times \albet^\ast\mid \exists t\in\Wfin{T}, u\in\albet^{+}.\,uw= u_t\},\]
	with the initial state $(\epsilon, \epsilon)$.
	The transition function is defined depending on the following cases (in both of them $a\in\albet$ is a single letter):
	\begin{align*}
	\delta((c,aw), t)&\eqdef \big(a, (c\cdot t, w)\big),&&\text{and}\\
	\delta((c,\epsilon), t)&\eqdef \big(a, (\epsilon, w)\big),&&\text{where $u_{c\cdot t}=aw$.}
	\end{align*}
	In this construction $c$ stores the product of those input letters which did not yet contributed to the output; initially $c=\epsilon$.
	After reading a~letter $t$, the transducer starts outputting the word $u_{c\cdot t}$, whose image under $\beta$ is $c\cdot t$.
	This word may consist of multiple letters, so the suffix that still needs to be outputted is stored on the second coordinate of the state,
	while the first coordinate is multiplied by the input letters read in the meantime.
	When the whole word is outputted, we read the next letter, and we start outputting the next word, corresponding to the product of letters read in the meantime; and so on.
\end{proof}

\begin{proof}[Proof of \cref{thm:transducer}]
	Recall that we already fixed an input $\omega$\=/word $\word{w}\in(\albet_W)^\omega$, as well as a homomorphism $\beta\from\Wilk{(\albet_Y)}\to T$ onto a~finite Wilke algebra $T$.
	Take the transducers $\theta_{\word{w}}\from \albet_W\times\albet_X\tto \Peps{\Wfin{T}}$ from \cref{lem:safety-delayed},
	$\theta_\mathsf{d}\from\Peps{\Wfin{T}}\tto\Wfin{T}$ from \cref{lem:elim-neutral},
	and $\theta_\beta\from \Wfin{T}\tto \albet_Y$ from \cref{lem:monoid2letter},
	where as the parameter $r\in\N$ in \cref{lem:safety-delayed} we take the constant $r$ provided by \cref{lem:elim-neutral}.
	As $\tau_{\word{w}}\from\albet_W\times\albet_X\to\albet_Y$ we take the transducer being the composition
	of the transducers $\theta_{\word{w}}$, $\theta_\mathsf{d}$, and $\theta_\beta$.
	Given now also a~word $\word{x}\in(\albet_X)^\omega$, we know that $\varphi(\word{w},\word{x},\theta_{\word{w}}(\word{w},\word{x}))$ holds;
	the equalities
	\[\odot\big(\theta_{\word{w}}(\word{w},\word{x})\big) = \odot\big(\theta_\mathsf{d}\big(\theta_{\word{w}}(\word{w},\word{x})\big)\big)
		=\beta\big(\theta_\beta\big(\theta_\mathsf{d}\big(\theta_{\word{w}}(\word{w},\word{x})\big)\big)
		=\beta\big(\tau_{\word{w}}(\word{w},\word{x})\big)\]
	imply that $\varphi(\word{w},\word{x},\tau_{\word{w}}(\word{w},\word{x}))$ holds as well.
\end{proof}

\section{Proof of \texorpdfstring{\cref{thm:trees-index-undec}}{Theorem \ref{thm:trees-index-undec}}}
\label{app:trees-index-undex}

\thmTressIndexUndec*

We first rely on results of Bojańczyk et al.~\cite{bojanczyk_msou_final} to argue that an~intermediate formalism MSO+\textsf{bounded} is undecidable. Then we prove \cref{lem:bound2safety}, showing that this formalism encodes within MSO+$\qI$.

\subsection{Undecidability of \texorpdfstring{MSO+\textsf{bounded}}{MSO+bounded}}

We rely on the proof that MSO+$\qU$ is undecidable, by Bojańczyk et al.~\cite{bojanczyk_msou_final}.
Structures considered in their paper are infinite sequences (indexed by natural numbers) of finite trees of depth\=/$4$.
Nodes in these trees may have arbitrarily many children, which are ordered, and each path from a~node to its descendant consists of at most $4$ nodes.
A~logic over such a~structure has access to binary relations ``being a child'', ``being the next sibling'', and ``being the root of the next tree''.

Let us now give a few definitions (with the same names as in the full binary tree, but with a~different meaning).
A set $I$ of nodes in a~sequence of depth\=/$4$ trees is called an~\emph{interval} if all its elements are consecutive siblings in some tree.
Two intervals $I_1$, $I_2$ are \emph{independent} if no node of one interval is a~descendant, ancestor, or sibling of a node of the other interval (in other words, the intervals need to either be contained in distinct trees, or their parents need to be incomparable with respect to the descendant order).
Finally, a \emph{union of independent intervals} is a~set $\tree{z}$ that can be written as $\bigcup_{i\in J}I_i$, where $(I_i)_{i\in J}$ are pairwise independent intervals.
Consider now an~extension of MSO, interpreted over sequences of depth\=/$4$ trees, in which we have access to an~atomic formula $\mathsf{bounded}(Z)$,
saying that $Z$ is a~union of independent intervals whose sizes are bounded by some $n\in\N$;
denote this extension by MSO+\textsf{bounded}.
We use the following theorem.

\begin{theorem}[\cite{bojanczyk_msou_final}]\label{thm:mso+u}
	The following problem is undecidable: given a~sentence $\psi$ of the logic MSO+\textsf{bounded}, say whether there exists a~sequence of depth\=/$4$ trees in which $\psi$ holds.
\end{theorem}

Strictly speaking, the undecidability result of Bojańczyk et al.~\cite{bojanczyk_msou_final} is given for the MSO+$\qU$ logic, but such a~modification can be easily obtained.
Namely, their undecidability proof provides a~reduction from an~undecidable problem (language emptiness for Minsky machines) to satisfiability of MSO+$\qU$.
Formulae of MSO+$\qU$ constructed in this reduction use the $\qU$ quantifier only in subformulae $\varphi_\mathsf{b}(Z)$ saying, for a~union of independent intervals $Z$,
that intervals in $Z$ have lengths bounded by some number.
Replacing each such subformula by the atom $\mathsf{bounded}(Z)$, we obtain an~equivalent formula of MSO+\textsf{bounded}.
After this modification, we obtain a~reduction from an~undecidable problem to satisfiability of MSO+\textsf{bounded},
proving undecidability of the latter.

\subsection{From \texorpdfstring{MSO+\textsf{bounded}}{MSO+bounded} to trees}

Our goal now is to observe that the problem from \cref{thm:mso+u} can be expressed in MSO+$\qI^\Indet_{\CW_{0,1}}$ and in MSO+$\qI^\Idet_{\CW_{0,1}}$
over the full binary tree.
To this end, we encode a~sequence of depth\=/$4$ trees as a set $S$ of nodes of the infinite binary tree, where nodes of trees in the sequence are in one-to-one correspondence with nodes of $S$.
The root of the $i$\=/th tree of the sequence is encoded as $\dR^i$.
Then, each tree is represented using the first\=/child next\=/sibling encoding: if a node is encoded as $u$, then its $i$-th child is encoded as $u\dL\dR^i$.
Thus, the question ``does there exist a sequence of depth\=/$4$ trees'' changes to ``does there exists a~set $S$ encoding a~sequence of depth\=/$4$ trees''.
Clearly, MSO over the binary tree can express the fact that $S$ is indeed a~valid encoding of some sequence of depth\=/$4$ trees.
Moreover, the sentence $\psi$ talking about a~sequence of depth\=/$4$ trees can be rewritten into a~sentence over the binary tree talking about the encoding $S$ of this sequence,
as long as MSO constructions are used.
It remains to translate the new atomic formula $\mathsf{bounded}(Z)$.
Here, the fact that $Z$ is a~union of independent intervals can be easily expressed, the only difficulty is in boundedness.
But we observe that siblings in depth\=/$4$ trees are encoded along a branch going right;
thus an~interval in a~depth\=/$4$ tree becomes encoded as an~interval in the binary tree,
and a~union of independent intervals in a~sequence of depth\=/$4$ trees becomes encoded as a~union of independent intervals in the binary tree.
In consequence, we can express the fact that the sizes of intervals in $Z$ are bounded by the formula $\varphi_\mathsf{b}(Z)$ from \cref{lem:bound2safety}.
Recall that the formula $\varphi_\mathsf{b}(Z)$ can be written both in MSO+$\qI^\Indet_{\CW_{0,1}}$ and in MSO+$\qI^\Idet_{\CW_{0,1}}$,
so (taking \cref{lem:bound2safety} as a~proviso) this proves undecidability of both these logics.

\subsection{Proof of \texorpdfstring{\cref{lem:bound2safety}}{Lemma \ref{lem:bound2safety}}}
\label{appp:bound2safety}

In this subsection we formally prove the remaining implication from the proof of \cref{lem:bound2safety}, namely that if intervals in $\tree{z}$ have unbounded lengths then $\varphi_\mathsf{b}(\tree{z})$ does not hold.

Recall that
\[\varphi_\mathsf{b}(Z)\, \equiv\, \forall W\subseteq Z.\ \forall Y\subseteq Z.\ \qI^D_{\CW_{0,1}}X.\ \underbrace{\big(X\subseteq W\land\forall x\in X.\,\exists y\in Y.\,x\preceq y\big)}_{\psi(W,Y,X)},\]
	and that as $\tree{w}$ we take the set of all topmost points of all intervals in $\tree{z}$.
	Below we construct the set $\tree{y}$.

	\begin{claim}
		We can find nodes $\epsilon=u_0\prec u_1\prec u_2\prec\dots$ and intervals $I_0,I_1,I_2,\dots$ from $\tree{z}$ such that $\len(I_i)\geq i$ and $u_i\preceq\topn(I_i)$ but $u_{i+1}\not\preceq\topn(I_i)$ for all $i\in\N$.
	\end{claim}

	\begin{claimproof}
		Suppose that we already have points $u_0,u_1,\dots,u_n$ and intervals $I_0,I_1,\dots,I_{n-1}$ as above, where intervals in $\tree{z}\cap\{u\mid u_n\preceq u\}$ have unbounded lengths
		(which we initialise by taking $u_0=\epsilon$ for $n=0$).
		We choose the next interval $I_n$ and node $u_{n+1}$ as follows.
		As $I_n$ we just take any interval in $\tree{z}\cap\{u\mid u_n\preceq u\}$ of length at least $n$.
		Then we consider nodes on the side of the path from $u_n$ to $\topn(I_n)$ (i.e., nodes that are not on that path, but whose parents are on that path).
		There are finitely many of them, so at least one such node $v$ should have the property that intervals in $\tree{z}\cap\{u\mid v\preceq u\}$ have unbounded lengths.
		We take any such $v$ as $u_{n+1}$.
	\end{claimproof}

	Below we consider non\=/deterministic safety automata over the alphabet $\{0,1\}^3$; they read binary trees with each node labelled by a triple $(w,y,z)\in\{0,1\}^3$,
	where $w$ says whether the node belongs to a set $\tree{w}$; likewise $y$ for $\tree{y}$ and $x$ for $\tree{x}$.
	We say that such an~automaton $\CA$ is \emph{monotone} if for every transition $(q,(w,y,1),0,q_\dL,q_\dR)$ it has also a~transition $(q,(w,y,0),0,q_\dL,q_\dR)$.
	In other words, if $\rho$ is a~run of $\CA$ over $\langle \tree{w},\tree{y},\tree{x}\rangle$, and $\tree{x}'\subseteq \tree{x}$, then $\rho$ remains a~valid run of $\CA$ over $\langle\tree{w},\tree{y},\tree{x}'\rangle$.

		\begin{claim}
		\label{cl:make-it-monotone}
		If for some $\tree{w}$, $\tree{y}$ the subformula $\psi(\tree{w},\tree{y},X)$ is recognised by some non\=/deterministic safety automaton $\CA$, then it is also recognised by a~monotone non\=/deterministic safety automaton $\CA'$.
		\end{claim}

		\begin{claimproof}
		To obtain $\CA'$, it is enough to add to $\CA$ all the transitions $\big(q,(w,y,0),0,q_\dL,q_\dR\big)$ such that $\big(q,(w,y,1),0,q_\dL,q_\dR\big)$ is a~transition of $\CA$.
		Then every run of $\CA$ is a~run of $\CA'$, so $\CA'$ accepts all sets $\tree{x}$ satisfying our property.
		Conversely, suppose that we have a~run of $\CA'$ for some set $\tree{x}'$.
		Then, by the construction of $\CA'$ there exists a~run of $\CA$ over some set $\tree{x}\supseteq \tree{x}'$ (whenever $\CA'$ used a~transition reading $(w,y,0)$ existing only in $\CA'$, we add the node to $\tree{x}$, and we use the corresponding transition of $\CA$ reading $(w,y,1)$).
		By the assumption on $\CA$, the set $\tree{x}$ satisfies our property, which is even more satisfied for its subset $\tree{x}'\subseteq \tree{x}$.
		\end{claimproof}

		Consider now a~list of all monotone non\=/deterministic safety automata $\CA_0,\CA_1,\CA_2,\ldots$ over the alphabet $\{0,1\}^3$.
		Fix some $n\in\N$ and suppose that we already have a~number $m_n$ and we already fixed the set $\tree{y}\cap\{u\mid u_{m_n}\not\preceq u\}\subseteq \tree{z}$
		(i.e.,~for nodes not being descendants of $u_{m_n}$ we have already decided which of them belong to $\tree{y}$ and which do not),
		with the property that no matter how we choose $\tree{y}\cap\{u\mid u_{m_n}\preceq u\}\subseteq \tree{z}$, the automata $\CA_0,\ldots,\CA_{n-1}$ do not recognise the subformula $\psi(\tree{w},\tree{y},X)$ (taking $m_0=0$ we trivially have the above property for $n=0$).
		We now define the next number $m_{n+1}>m_n$, and the next fragment of the set $\tree{y}$, namely $\tree{y}\cap\{u\mid u_{m_n}\preceq u\land u_{m_{n+1}}\not\preceq u\}$, in a~way that the above holds for $n{+}1$.
		After finishing this construction for all $n\in\N$,
		we will obtain a~whole set $\tree{y}\subseteq \tree{z}$ such that no monotone non\=/deterministic safety automaton $\CA$
		correctly recognises the subformula~$\psi(\tree{w},\tree{y},X)$. Thus, \cref{cl:make-it-monotone} implies that no non\=/deterministic (thus neither deterministic) safety automaton recognises the subformula~$\psi(\tree{w},\tree{y},X)$, meaning that $\varphi_\mathsf{b}(\tree{z})$ does not hold.

		For an~automaton $\CA_n$, with set of states $Q$, we proceed as follows.
		To each pair $(i,j)$ with $i<j$ we assign the set $R_{i,j}$ of triples $(q_\top,q_\bot,q_I)$ such that
		there is an~accepting (i.e.,~without transitions of priority $1$) run of $\CA_n$ over the context $C_{i,j}=\{u\mid u_i\preceq u\land u_j\not\preceq u\land\topn(I_i)\not\preceq u\}$,
		where both $\tree{y}$ and $\tree{x}$ are empty (and where $\tree{w}$ contains topmost nodes of all intervals in $\tree{z}$), with $q_\top$ in the root $u_i$, with $q_\bot$ in the hole $u_j$, and with $q_I$ in the hole $\topn(I_i)$.
		By the Ramsey theorem (see \cref{thm:ramsey}), we can find indices $i_0,\ldots,i_{2|Q|}$ with $\max(m_n,|Q|+1)\leq i_0<i_1<\ldots<i_{2|Q|}$ such that all $R_{i_j,i_{j'}}$ for $0\leq j<j'\leq 2|Q|$ are the same.
		We fix $m_{n+1}=i_{2|Q|}$, and we add the bottommost points of the intervals $I_{i_0},I_{i_2}, I_{i_4},\ldots,I_{2|Q|}$ (even numbers only!) to $\tree{y}$.

		Suppose now that the rest of $\tree{y}$ (i.e.,~its part below $u_{m_{n+1}}$) is also fixed.
		Take $\tree{x}$ containing topmost points of the intervals $I_{i_0},I_{i_2},I_{i_4},\ldots,I_{2|Q|}$ (again even numbers only).
		For such a set $\tree{x}$ the subformula $\psi(\tree{w},\tree{y},\tree{x})$ holds.
		If $\CA_n$ does not accept the triple $\langle \tree{w},\tree{y},\tree{x}\rangle$, we are done.
		Suppose it accepts, and fix some accepting run $\rho$.
		By the pigeonhole principle, there are indices $b$, $e$ with $0\leq b<e\leq|Q|$ such that $\rho(u_{i_{2b}})=\rho(u_{i_{2e}})$; denote this state $q$.
		Denote also $q_1=\rho(u_{i_{2b+2}})$.

		\begin{claim}\label{cl:1}
			We have $\big(q,q_1,\rho(\topn(I_{i_{2b}}))\big)\in R_{i_{2b},i_{2b+2}}=R_{i_{2b+1},i_{2b+2}}$.
		\end{claim}

		\begin{claimproof}
			This is directly witnessed by the run $\rho$, because our sets $\tree{y}$ and $\tree{x}$ do not have any elements in the context $C_{i_{2b},i_{2b+2}}$.
			The equality is by the definition of the indices $i_j$.
		\end{claimproof}

		\begin{claim}\label{cl:2}
			For every $i\geq|Q|+1$ (in particular for all $i=i_{2j}$ with $b<j<e$) there is a~run of $\CA$ over the subtree starting in $\topn(I_i)$,
			with the original $\tree{w}$ (containing only $\topn(I_i)$), with $\emptyset$ taken as $\tree{y}$ and $\tree{x}$, and with the state $\rho(\topn(I_i))$ in $\topn(I_i)$.
		\end{claim}

		\begin{claimproof}
			First, by monotonicity of $\CA_n$ we can treat $\rho$ reading $\langle\tree{w},\tree{y},\tree{x}\rangle$ as a~run reading $\langle\tree{w},\tree{y},\emptyset\rangle$.
			Note that $\tree{y}$ restricted to our subtree can only contain the bottommost node of $I_i$ (some node on the rightmost branch, in distance at least $|Q|+1$ from the root).
			By the pigeonhole principle there are two indices $r$, $s$ with $1\leq r<s\leq|Q|+1$ with $\rho(\topn(I_i)\dR^r)=\rho(\topn(I_i)\dR^s)$.
			Then, below $\topn(I_i)\dR^r$ we repeat forever the part of the run from the context between $\topn(I_i)\dR^r$ and $\topn(I_i)\dR^s$,
			obtaining a~run over a~tree with empty $\tree{y}$ and $\tree{x}$ (and because $r\geq 1$ we do not copy the root node contained in $\tree{w}$).
		\end{claimproof}

		\begin{claim}\label{cl:3}
			For every $i\geq|Q|+1$ (in particular for $i=i_{2b}$) there is an~accepting run of $\CA_n$ over the subtree starting in $\topn(I_i)$,
			with the original $\tree{w}$ and $\tree{x}$, and with $\emptyset$ taken as $\tree{y}$, which has the state $\rho(\topn(I_i))$ in $\topn(I_i)$.
		\end{claim}

		\begin{claimproof}
			As above, but we do not make $\tree{x}$ empty.
			Note that $\tree{x}$ may contain only the root of the subtree, so elements of $\tree{x}$ are not copied.
		\end{claimproof}

		\begin{claim}\label{cl:4}
			We have $\big(q,q,\rho(\topn(I_{i_{2b}}))\big)\in R_{i_{2b},i_{2e}}=R_{i_{2b},i_{2b+1}}$.
		\end{claim}

		\begin{claimproof}
			We modify the run $\rho$ over the context $C_{i_{2b},i_{2e}}$.
			Originally, this context contains some elements of $\tree{y}$ and $\tree{x}$ in subtrees starting in $\topn(I_{i_{2j}})$ for $b<j<e$,
			but we make them empty using \cref{cl:2}.
		\end{claimproof}

		We modify the original run $\rho$ as follows.
		\begin{itemize}
		\item	On the context $C_{i_{2b},i_{2b+1}}$ we replace the original run by any run with $q$ in $u_{i_{2b}}$ (as previously), with $q$ in $u_{i_{2b+1}}$, and with $\rho(\topn(I_{i_{2b}}))$ in $\topn(I_{i_{2b}})$ (as previously),
			existing by \cref{cl:4} and by the definition of $R_{i_{2b},i_{2b+1}}$.
		\item	On the context $C_{i_{2b+1},i_{2b+2}}$ we replace the original run by any run with $q$ in $u_{i_{2b+1}}$ (agrees with the part defined above), with $q_1$ in $u_{i_{2b+2}}$ (as previously),
			and with $\rho(\topn(I_{i_{2b}}))$ in $\topn(I_{i_{2b+1}})$, existing by \cref{cl:1} and by the definition of $R_{i_{2b+1},i_{2b+2}}$.
		\item	We add $\topn(I_{i_{2b+1}})$ to $\tree{x}$, thus we define $\tree{x}'\eqdef \tree{x}\cup\{\topn(I_{i_{2b+1}})\}$.
		\item	The subtree starting in $\topn(I_{i_{2b+1}})$ differs from the subtree starting in $\topn(I_{i_{2b}})$ only in the fact that there are no elements of $\tree{y}$ in the former subtree.
			Thus, we can cover it by a~run starting with $\rho(\topn(I_{i_{2b}}))$ in $\topn(I_{i_{2b+1}})$, by \cref{cl:3}.
		\end{itemize}
		This way we obtain an~accepting run of $\CA_n$ on $\langle\tree{w},\tree{y},\tree{x}'\rangle$.
		But no node below the new node of $\tree{x}'$, namely $\topn(I_{i_{2b+1}})$, belongs to $\tree{y}$, so the subformula $\psi(\tree{w},\tree{y},\tree{x}')$ does not hold.
		This means that $\CA_n$ does not recognise correctly the subformula $\psi(\tree{w},\tree{y},X)$.

	This concludes the inductive construction of $\tree{y}$ and thus shows that $\varphi_\mathsf{b}(\tree{z})$ fails when the intervals in $\tree{z}$ have unbounded lengths. Thus, the proof of \cref{lem:bound2safety} is finished.
\end{document}

%% file: macros.tex
\usepackage{amssymb} % collection de symboles mathématiques

\usepackage[dvipsnames]{xcolor}

\usepackage{pifont}
\usepackage{todonotes}
\usepackage{mathtools}
\usepackage{stmaryrd}
\usepackage{xspace}

\usepackage[shortcuts]{extdash}

\usepackage{tikz}

\newcommand{\ident}[1]{\ensuremath{\texorpdfstring{\mathrm{\mathnm{#1}}}{#1}}\xspace}

\newcommand{\domain}[1]{\ensuremath{\mathbb{#1}}\xspace}

\newcommand{\ot}{\leftarrow}

\newcommand{\qU}{\ensuremath{\mathsf U}\xspace}
\newcommand{\qG}{\rotatebox[origin=c]{180}{$\mathtt{G}$}\xspace}
\newcommand{\qI}{\mathtt{I}}
\newcommand{\qQ}{\mathtt{Q}}

\newcommand{\albet}{A\xspace}

\newcommand{\theGame}{\CG}
\newcommand{\Idet}{\mathsf{dt}}
\newcommand{\Indet}{\mathsf{nd}}
\newcommand{\sto}{{\mathrel{\mkern-1mu\scalebox{0.7}[1.0]{$\shortrightarrow$}\mkern-1mu}}}

\newcommand{\powerset}{\mathsf{\mathnm{P}}}

\newcommand{\N}{\domain{N}}

\newcommand{\eqdef}{\stackrel{\text{def}}=}

\renewcommand{\mod}[1]{\allowbreak\mkern10mu({\operator@font mod}\,\,#1)}

\newcommand{\rI}{\mathrm{I}}
\newcommand{\rII}{\mathrm{II}}

\newcommand{\mathcalsym}[1]{\ensuremath{\mathcal{#1}}\xspace}

\newcommand{\CA}{\mathcalsym{A}}
\newcommand{\CB}{\mathcalsym{B}}
\newcommand{\CC}{\mathcalsym{C}}
\newcommand{\CD}{\mathcalsym{D}}

\newcommand{\CG}{\mathcalsym{G}}

\newcommand{\CM}{\mathcalsym{M}}

\newcommand{\CP}{\mathcalsym{P}}

\newcommand{\CR}{\mathcalsym{R}}

\newcommand{\CW}{\mathcalsym{W}}

\newcommand{\from}{\colon}

\newcommand{\lang}{\ident{L}}

\usetikzlibrary{positioning}
\usetikzlibrary{decorations.pathreplacing}
\usetikzlibrary{automata,positioning}
\usetikzlibrary{decorations.pathmorphing}
\usetikzlibrary{decorations.markings}
\usetikzlibrary{decorations}
\usetikzlibrary{arrows}
\usetikzlibrary{patterns}
\usetikzlibrary{calc}
\usetikzlibrary{shapes}
\usetikzlibrary{fadings,	shadings}

\usetikzlibrary{arrows.meta}
	\tikzstyle{ubrace} = [draw, thick, decoration={brace, amplitude=7pt, mirror, raise=0.0cm}, decorate,
	    every node/.style={anchor=north, yshift=-0.15cm}]
	\tikzstyle{rbrace} = [draw, thick, decoration={brace, amplitude=7pt, mirror, raise=0.0cm}, decorate,
	    every node/.style={anchor=west, xshift= 0.15cm}]

	\tikzstyle{obrace} = [draw, thick, decoration={brace, amplitude=7pt, raise=0.0cm}, decorate,
	    every node/.style={anchor=south, yshift= 0.15cm}]
	\tikzstyle{lbrace} = [draw, thick, decoration={brace, amplitude=7pt, raise=0.0cm}, decorate,
	    every node/.style={anchor=east, xshift=-0.15cm}]

\tikzstyle{bzFlow} = [] % [inner sep=0cm, node distance=0cm and 0cm, outer sep=0cm]

\tikzstyle{bzL} = [bzFlow, anchor=mid west, align=left]

\tikzstyle{bzC} = [bzFlow, anchor=mid     , align=center]

\tikzstyle{bzR} = [bzFlow, anchor=mid east, align=right]

\newcommand{\dL}{\mathtt{{\scriptstyle L}}}
\newcommand{\dR}{\mathtt{{\scriptstyle R}}}

\newcommand{\mathnm}[1]{#1}

\newcommand{\tran}[1]{\xrightarrow{#1}}

\newcommand{\restr}{{\upharpoonright}}

\tikzstyle{dot} = [draw,shape=circle,fill, minimum size=1mm, inner sep=0pt, outer sep=0pt]
\tikzstyle{edge}  = [draw, thick,->]

\newcommand{\init}{\ident{I}}

\DeclareSymbolFont{yhlargesymbols}{OMX}{yhex}{m}{n}
\DeclareMathAccent{\yhwidehat}{\mathord}{yhlargesymbols}{"62}

\newcommand{\trees}{\ident{Tr}}

\newcommand{\fcirc}{\mathbin{\vcenter{\hbox{\scalebox{0.8}{$\bullet$}}}}}

\newcommand{\word}[1]{\bar{#1}}
\newcommand{\tree}[1]{\tilde{#1}}
\newcommand{\funapp}[2]{#1\fcirc #2}
\newcommand{\Peps}[1]{#1_{+\epsilon}}

\newcommand{\Wfin}[1]{{#1}^{\mathrm{fin}}}
\newcommand{\Winf}[1]{{#1}^{\mathrm{inf}}}
\newcommand{\Wilk}[1]{{#1}^{\mathsf{W}}}

\newcommand{\anIndex}{\CR\xspace}

\renewcommand\phi\varphi
\renewcommand\epsilon\varepsilon

\newcommand{\ignore}[1]{}

\DeclarePairedDelimiter{\set}{\{}{\}}

\newcommand\ms[1]{\todo[inline, size=\scriptsize,backgroundcolor=Turquoise]{#1 - \textbf{Michał}}}

\Crefname{lemma}{Lemma}{Lemmata}
\Crefname{fact}{Fact}{Facts}
\Crefname{claim}{Claim}{Claims}
\Crefname{equation}{Formula}{Formulae}
\Crefname{property}{Property}{Properies}\creflabelformat{property}{(#2#1#3)}

\newcommand{\tto}{\twoheadrightarrow}